\documentclass{article}
\usepackage[utf8]{inputenc}
\usepackage{amsmath}
\usepackage{amsthm}
\usepackage{amssymb}
\usepackage{amsfonts}
\usepackage{mathtools}
\usepackage{breakcites}
\usepackage{mwe}
\usepackage{paralist, tabularx}
\usepackage{enumitem}
\setlist[itemize]{noitemsep, topsep=3pt}
\setlist[enumerate]{noitemsep, topsep=3pt}

\newcommand{\tabincell}[2]{\begin{tabular}{@{}#1@{}}#2\end{tabular}}

\usepackage{graphicx}
\usepackage{xcolor}
\usepackage{algpseudocode}
\usepackage{algorithm, algpseudocode}
\usepackage[noend]{algcompatible}
\usepackage{paralist, tabularx, natbib}
\usepackage[colorlinks,
    linkcolor={red!50!black},
    citecolor={blue!50!black},
    urlcolor={blue!80!black}]{hyperref}
\usepackage[margin=1in]{geometry}

\usepackage{caption}
\usepackage{subcaption}
\usepackage{booktabs} 

\usepackage{fullpage}

\newtheorem{thm}{Theorem}

\newtheorem{lemma}[thm]{Lemma}

\newtheorem{proposition}[thm]{Proposition}

\newcommand{\cev}[1]{\reflectbox{\ensuremath{\vec{\reflectbox{\ensuremath{#1}}}}}}

\newtheorem{assumption}[thm]{Assumption}
\newcommand*\samethanks[1][\value{footnote}]{\footnotemark[#1]}

        {\hspace*{\fill}$\qed$\par}

\usepackage{color}

\title{Predictive Flows for Faster Ford-Fulkerson}

\begin{document}

\author{Sami Davies\thanks{Northwestern University. \texttt{sami@northwestern.edu}. Supported by an NSF Computing Innovation Fellowship.}, Benjamin Moseley\thanks{Carnegie Mellon University. \texttt{moseleyb@andrew.cmu.edu}. Supported in part by a Google Research Award, an Inform Research Award, a Carnegie Bosch Junior Faculty Chair, and NSF grants CCF-2121744 and  CCF-1845146.}, Sergei Vassilvitskii\thanks{Google Research -- New York. \texttt{sergeiv@google.com}, \texttt{wangyy@google.com}.}, Yuyan Wang\samethanks}

\maketitle

\begin{abstract}
    
Recent work has shown that leveraging learned predictions can improve the running time of algorithms for bipartite matching and similar combinatorial problems.  In this work, we build on this idea to improve the performance of the widely used Ford-Fulkerson algorithm for computing maximum flows by seeding Ford-Fulkerson with predicted flows. 
Our proposed method offers strong theoretical performance in terms of the quality of the prediction.
We then consider image segmentation,  a common use-case of flows in computer vision, and complement our theoretical analysis with strong empirical results. 
\end{abstract}



\section{Introduction}
The Ford-Fulkerson method is one of the most ubiquitous in combinatorial optimization, 
both in theory and in practice. 
While it was first developed for solving the maximum flow problem, 
many problems in scheduling \citep{ahuja1993network}, computer vision \citep{vineet2008cuda}, 
resource allocation, matroid intersection \citep{ImMP21}, and other areas 
are solvable by finding a reduction to a flow problem.

Theoretically, max flow algorithms exist that are asymptotically much faster than 
the original Ford-Fulkerson formulation, 
most recently the near-linear time algorithm of \citet{KLPPS22}. 
However---as often happens---algorithms with great theoretical guarantees  
might be difficult to implement in practice.  Indeed, algorithms used in practice still leave room for improvement. 
In fact, for computing max flows in networks, 
practitioners often stick to older 
algorithms such as Dinic's algorithm \citep{BhadraKK20}, 
push-relabel \citep{cherkassky1997implementing}, pseudoflow \citep{chandran2009computational},
or these algorithms layered with heuristics to fit particular settings.

When flow algorithms are deployed in practice, 
they are often used to solve several problem instances arising naturally over time.  However, the theoretical analysis, as well as many implementations, 
considers solving each new problem from scratch to derive worst-case guarantees. 
This approach needlessly discards information that may exist between instances. 
We are interested in discovering whether flow problems can be solved more efficiently by leveraging information from past examples. 
Seeding an algorithm from a non-trivial starting point is referred to as a \textbf{warm-start}. 

We are motivated by the question:
can one warm-start Ford-Fulkerson to improve theoretical and empirical performance?
Towards this goal, we leverage the recently developed algorithms with predictions framework 
(a.k.a learning-augmented algorithms). 
Research over the past several years has showcased the power of
augmenting an algorithm with a learned prediction, 
leading to improvements in caching \citep{Im0PP22, LindermayrM22, LykourisV21},
scheduling \citep{Im0QP21, LattanziLMV20}, clustering \citep{LMVWZ21}, 
matching \citep{ChenSVZ22, DinitzILMV21}, and more (see the survey by 
\citet{MitzenmacherV22}). 
An algorithm is \textbf{learning-augmented} if it can use
a prediction that relays information about the problem instance.
Most prior  work  uses
predictions to overcome uncertainty in  the online setting. 
However, recent work by \citet{DinitzILMV21}---and the follow-up work by \citet{ChenSVZ22}---instead 
focuses on improving the run-time of bipartite matching algorithms by predicting the dual variables and using these to warm-start the primal-dual algorithm.

Motivated by the idea of warm-starting combinatorial optimization algorithms, we seek to provide faster run-time guarantees for flow problems via warm-start. 
The paper will focus on flow problems generally, 
but will additionally showcase a common, practical use-case in 
computer vision: image segmentation. 
In the image segmentation problem, the input is an image containing an object/ foreground,
and the goal is to locate the foreground in the image.

\subsection{Our contributions}

For a graph $G=(V,E)$ equipped with a capacity vector $c \in \mathbb{Z}_{\geq 0}^{|E|}$, let $f$ be a flow on $G$, where $f_e$ is the flow value on each edge $e \in E$. Let $\mathcal{F}^*$ be the collection of all feasible, maximum flows on $G$. Given a potentially infeasible flow $\widehat{f}$, let 
$\eta(\widehat{f}) = \min_{f^* \in \mathcal{F}^*}  || \widehat{f} - f^* ||_1$. This term denotes how close $\widehat{f}$ is to being optimal.



\noindent \textbf{Warm-starting Ford-Fulkerson on general networks\quad} 
Our main contribution is Algorithm \ref{alg:EK-warmstart}, which can be used to warm-start any implementation of Ford-Fulkerson, i.e., Ford-Fulkerson with any specified subroutine for finding augmenting paths. Algorithm \ref{alg:EK-warmstart} takes as input a predicted flow $\widehat{f}$. Note $\widehat{f}$ may be infeasible for $G$, as predictions can be erroneous. Algorithm \ref{alg:EK-warmstart} first projects $\widehat{f}$ to a feasible flow for $G$, and then runs the Ford-Fulkerson procedure from the feasible state to find a maximum flow. 
While our warm-started Ford-Fulkerson has its performance tied to the quality of the prediction, it also enjoys the same worst-case run-time bounds as the vanilla Ford-Fulkerson procedure.

\begin{thm} \label{thm: run-time-flow-0}
Let $\widehat{f}$ be a potentially infeasible flow on network $G=(V,E)$. Let $T$ be the worst-case run-time for Ford-Fulkerson with a chosen augmenting path subroutine. Using the same subroutine, Algorithm \ref{alg:EK-warmstart} seeded with $\widehat{f}$ finds an optimal flow $f^*$ on $G$ within time $ O(\min \{ |E|\cdot \eta(\widehat{f}), T\})$.
\end{thm}


At various points, we specify two Ford-Fulkerson implementations: Edmonds-Karp and Dinic's algorithm, for which the run-time $T$ is $O(|E|^2|V|)$ and $O(|V|^2|E|)$, respectively.

One may wonder how to obtain such a $\widehat{f}$.
We prove that when the networks come from a fixed but unknown distribution,
one can PAC-learn the best approximation (Theorem \ref{thm: learning-PAC}) for optimal flows, which can be used to warm-start.


\noindent \textbf{Faster warm-start on locally-changed networks\quad} 
Next, we improve the analysis of Algorithm \ref{alg:EK-warmstart} for network instances with gradual, local transitions from one to another. 
We prove the local transitions among networks, informally characterized in Theorem \ref{thm: instance-robust-0},
give rise to many short paths along which we can send flow, thus improving the run-time.

\begin{thm}[Informal, formally Theorem \ref{thm:instance-robustness-formal}]
\label{thm: instance-robust-0}
Fix separable networks $G^1$ and $G^2$, 
where the transition between 
them is $d$-local.
For $\widehat{f}$ an optimal flow on $G^{1}$, 
the run-time of Algorithm \ref{alg:EK-warmstart} seeded with $\widehat{f}$ on 
$G^{2}$ to find optimal $f^*$ on $G_2$ is
$O( d \cdot |E| +d^2\cdot \eta(\widehat{f}))$.
\end{thm}

\noindent \textbf{Empirical results\quad}  
Motivated by our theoretical results, we use 
our warm-started Ford-Fulkerson procedure on networks derived from instances of \textit{image segmentation} on sequences of photos taken of a moving object or from changing angles. 
We show that warm-start is faster than standard Ford-Fulkerson procedures (2-$5\times$ running time improvements), thus demonstrating that our theory is predictive of practical performance.  
A key piece of the speed gain of warm-start comes not from sending the flow along fewer paths, but rather from using shorter paths to project $\widehat{f}$ to a feasible flow, 
as predicted by Theorem \ref{thm: instance-robust-0}.
We note that the goal of our experiments is not necessarily to provide state-of-the-art algorithms for image segmentation, but instead to show that
warm-starting Ford-Fulkerson leads to substantial run-time improvements on practical networks as compared to running Ford-Fulkerson from scratch.

\subsection{Related work}

Flow problems have been well studied.  See the survey by \citet{ahuja1993network}.  
The Ford-Fulkerson method greedily computes a maximum flow by iteratively using an augmenting-path finding subroutine \citep{ford_fulkerson_1956}. Different subroutines give rise to different implementations such as Edmonds-Karp (using BFS) \citep{EmondsK} and the even faster Dinic's algorithm \citep{dinitz2006dinitz}. 
\citet{Sherman13} and \citet{KelnerLOS14} give fast algorithms that compute approximate maximum flows. \citet{KLPPS22} gave a nearly-linear time max flow algorithm.

Similar in spirit to our work, 
\citet{AltnerE08} demonstrate empirically that one can warm-start the push-relabel algorithm
on  similar networks.
Additionally, we are aware of concurrent work on
a warm-started max flow algorithm by \citet{PolakZub}. Importantly, they require an additional assumption that the predicted flow satisfy flow conservation constraints, a limitation that the authors highlight. In contrast, we have an explicit feasibility restoration step, allowing us to get rid of this assumption.

Learning-augmented algorithms
have become popular recently.  
The area was jump started by \citet{KraskaBCDP18},
who showed results on learned data structures.  
The area has since become popular for the design of online algorithms 
where the algorithm uses predictions to cope with uncertainty \citep{PurohitSK18,LattanziLMV20,AntoniadisGKK20}.
For the reader less familiar with this literature, 
we recommend reading the paragraph \emph{Learning-augmented algorithms} in Section \ref{sec: prelims}.

While \citet{KraskaBCDP18} showed that running times can be improved using predictions,
this is still yet to be well-understood theoretically.  
The work of ~\citet{DinitzILMV21} showed how to improve the run-time of the Hungarian algorithm for weighted bipartite matching. 
~\cite{ChenSVZ22} has extended this to other graph problems.  
Both of these works warm-start
primal dual algorithms with a predicted dual solution.
Other run-time improvements have been made using predicitions too, 
as \citet{SakaueO22} gave algorithms for faster discrete convex analysis, and
\citet{LuRSZ21} showed predictions can be used to improve the run-time of generalized sorting. 
A closely related area is that of data-driven algorithm design \citep{GuptaR17,BalcanDDKSV21}.

\subsection{Organization and preliminaries}
\label{sec: prelims}

Section \ref{sec: FF} presents the warm-start  algorithm, proves its correctness, and provides run-time guarantees. In Section \ref{sec: instance-robust}, 
we give additional theoretical guarantees 
if the networks are from a specific subclass. 
In Section \ref{sec: empirical}, we show our empirical results on networks arising from image segmentation, which are closely related to the networks in Section \ref{sec: instance-robust}. Additional empirical evidence can be found in Appendix \ref{sec: appendix} and more at \url{https://github.com/wang-yuyan/warmstart-graphcut-algorithms-pulic}.

\noindent \textbf{Learning-augmented algorithms\quad}
In this model, an algorithm is given access to a predicted parameter.  The prediction can be erroneous and  must come equipped with an error metric.
In our setting, for a given network $G$, we will predict an integral flow $\widehat{f}$. 
This predicted flow may be infeasible for $G$. 
Recall that for $\mathcal{F}^*$ the set of optimal flows on $G$, 
we define the error of $\widehat{f}$ on $G$ to be
$\eta(\widehat{f}) = \min_{f^* \in \mathcal{F}^*}||\widehat{f}-f^* ||_1$.

It is well-established in the literature on learning-augmented algorithms that the 
desired properties of the prediction and algorithm
are \emph{learnability, consistency, and robustness}.
We show that given an additional assumption on the uniqueness of optimal flows (see Assumption \ref{asmp: unique}), 
predicted flows are PAC-learnable in Theorem \ref{thm: learning-PAC}.
Additionally, we use the \textit{instance-robustness} 
\citep{LavastidaM0X21} of flows 
to justify learnability for special networks in Theorem \ref{thm: instance-robust-0} 
and observe this empirically, as well.
If we are given a predicted flow $\widehat{f}$ for a 
network $G$ that is actually an optimal flow, then the run-time is 0, 
so Algorithm \ref{alg:EK-warmstart} is consistent.
We are also guaranteed robustness and a worst-case guarantee, 
as the run-time of Algorithm \ref{alg:EK-warmstart} is 
$O(\min\{|E|\eta(\widehat{f}),T\})$, 
which degrades smoothly as a function of how far $\widehat{f}$ is far from an optimal flow $f^*$ on a network,
but the worst-case is still bounded by $O(T)$\footnote{Recall, $T$ depends on the Ford-Fulkerson implementation.}.

\noindent \textbf{Network flow}\quad
Let $G = (V , E)$ be a fixed network with $|V|=n$ and $|E|=m$.
The source $s$ and sink $t$ are part of the vertex set $V$.
The network is equipped with a capacity vector $c \in \mathbb{Z}_{\geq 0}^{m}$. 
A flow $f \in \mathbb{Z}_{\geq 0}^m$ on $G$ is feasible if it satisfies flow conservation,
i.e. for all vertices $u \in V \setminus \{s,t\}$ the incoming flow for $u$ equals the outgoing flow 
$\sum_{e = (v,u)}f_e = \sum_{e = (u,w)}f_e$, and capacity constraints,
i.e.  $f_e \leq c_e$ for all edges $e \in E$. Throughout the paper we refer to a flow satisfying these constraints as \textbf{feasible}.

Given a flow $f$ that satisfies capacity constraints but not necessarily flow conservation,
the \textbf{residual graph} $G_f$ is the network on the same set of vertices,
and for any edge $e=(u,v) \in E$, add edge $e$ to $G_f$ but with capacity $c'_e=c_e-f_e$ 
and a reversed edge $\cev{e}=(v,u)$ with capacity $f_e$.
Let $\nu(f)$ be the amount of flow that $f$ sends from the source, $\nu(f) = \sum_{e =(s,u)} f_e$. 
An \textbf{augmenting path} in $G_f$ is a path $p$ from $s$ to $t$ where every edge $e \in p$ has $c'_e>0$.

When $f$ does not satisfy flow conservation, the total incoming and outgoing flow values on a node are different. Call this difference the \textbf{excess}/\textbf{deficit} ($\textsf{ex}_f/ \textsf{def}_f$) of the node if the incoming is more/less than outgoing flow, respectively.
For shorthand, we let the total excess and deficit in $G$ according to flow $f$ be
$\textsf{ex}_f = \sum_{u \not \in \{s,t\}} \textsf{ex}_{f}(u)$
and $\textsf{def}_f = \sum_{u \not \in \{s,t\}} \textsf{def}_{f}(u)$;
note that this excludes the source and sink.
Let $A_f, B_f \subseteq V$ be the nodes with positive excess/deficit with respect to $f$ (so $t \in A_f$, $s \in B_f$), respectively. It will be convenient to further refer to these sets excluding $s, t$ with
$A'_{f} = A_{f} \setminus \{t\}$ and $B'_{f} = B_{f} \setminus \{s\}$.

\section{Warm-start Ford-Fulkerson}
\label{sec: FF}
Here, we give our algorithm for using predicted flows.  
The next proposition holds from the following observations. 
Any Ford-Fulkerson method (e.g., Edmonds-Karp or Dinic's) 
can be seeded with any feasible flow, 
since running Ford-Fulkerson seeded with a feasible flow is the same as running it from scratch on the residual graph. Each iteration of Ford-Fulkerson increases the value of the flow and takes $O(|E|)$ time to find an augmenting path and send flow.
\begin{proposition}
\label{prop: feasibile-ws}
    Let $f$ be a feasible flow on $G$, 
    where $\nu(f) < \nu(f^*)$ for $f^*$ an optimal flow on $G$. 
    Ford-Fulkerson seeded with $f$ terminates in at most $\nu(f^*)-\nu(f)$ many iterations, 
    so its run-time is $O(|E| \cdot (\nu(f^*)-\nu(f)))$.
\end{proposition}

Let $\widehat{f}$ be a predicted flow for network $G$.
It may be infeasible, that is, it can violate capacity or flow conservation constraints. 
Algorithm \ref{alg:EK-warmstart} has two primary steps: 
step (1) projects $\widehat{f}$ to a feasible flow---we call this the \textbf{feasibility projection}---and step (2) runs a Ford-Fulkerson method seeded with a feasible flow and finds an optimal flow.
During feasibility projection, we first round down the flow wherever capacity constraints are violated. Then we send flow along \textbf{projection paths}, that is, a path from excess nodes to deficit nodes where all capacities are positive in the residual graph, to get flow conservation.

\iftrue
\begin{algorithm}[t] \label{alg:ws}
\caption{Warm-starting Ford-Fulkerson with $\widehat{f}$}\label{alg:EK-warmstart}
\begin{algorithmic}

\STATE {\bfseries Input}: predicted flow $\widehat{f}$
\WHILE{ $\exists$ edge $e$ in $G$ with $\widehat{f}_e > c_e$}
\STATE Update $\widehat{f}$ to $\widehat{f}_e \leftarrow c_e$
\ENDWHILE
\STATE Set $f \leftarrow \widehat{f}$
\STATE Build the residual graph $G_{f},$ as well as $A'_f$ and $B'_f$, the sets of nodes with excess/deficit to round
\STATE \hspace{-4mm}  // \emph{Main \textsf{while} loop, feasibility projection}
\WHILE{ $ |A'_f \cup B'_f| > 0$ }
\IF{$|A'_f| > 0$}
\IF{$\exists$ projection path from $u \in A'_f$ to $v \in B'_f$}
\STATE Let $p$ be the path from $u$ to $v$
\ELSE
\STATE Let $p$ be a path from $u \in A'_f$ to $s$
\ENDIF
\ELSE 
\STATE Let $v \in B'_f$, let $p$ be a path from $t$ to $v$ 
\STATE \hspace{-4mm}  // \emph{See the text for more details on path-finding}
\ENDIF
\STATE Let $w, w'$ be the beginning and ending nodes of $p$, respectively
\STATE Send $\mu_p = \min\{\textsf{ex}_{f}(w), \textsf{def}_{f}(w'), \min_{e \in p}c'_e\}$ units of flow down path $p$ 
\State Update $f$, $G_f,$ $A'_f$, and $B'_f$
\ENDWHILE
\STATE \hspace{-4mm}  // \emph{Feasibility projection ends, optimization starts}
\STATE Run Ford-Fulkerson on $G$ seeded with $f$ until optimality
\STATE {\bfseries Output}: $f^*$
\end{algorithmic}
\end{algorithm}
\fi

Let $f$ be an infeasible flow that satisfies capacity constraints and is integral. 
In the main \textsf{WHILE} loop of Algorithm \ref{alg:EK-warmstart}, projection paths are found in three rounds: $A'_f - B'_f$, $A'_f - s$, $B'_f - t$. Within each rounds, we find the projection paths by constructing an auxiliary graph $G'$ and applying on this graph the chosen augmenting path subroutine (e.g., BFS) in Ford-Fulkerson. 
To build $G'$ for finding all possible $A'_f - B'_f$ projection paths, take the residual graph $G_f$ w.r.t $f$ and treat it as a new network. Add a super source $s^*$ and super sink $t^*$. Add arcs $(s^*,u)$ to every $u \in A_f'$ (non-sink excess nodes) with capacity $\textsf{ex}_f(u)$, and $(u,t^*)$ from every $u \in B_f'$ (non-source deficit nodes) to $t^*$. Initialize all flows to be $0$ on this network. An $s^* - t^*$ augmenting path in $G'$ corresponds to a $A'_f - B'_f$
projection path in $G_f$. This is because for any projection path from $u$ to $u'$ one can let there be flow on $s^*$ to $u$ and $u'$ to $t^*$, thus making it an augmenting path in $G'$, and the reverse procedure holds. The $G'$ graphs for the other two rounds are built similarly; for finding $A'_f - s$ projection paths, add arcs $(s^*, u)$ to $u \in A_f'$ and $(s, t^*)$; for $t - B'_f$, add arcs $(s^*, t)$ and $(v, t^*)$ to every $v \in B_f'$.

To see that in the main \textsf{WHILE} loop of Algorithm \ref{alg:EK-warmstart}, the projection paths can be found from $A'_f - B'_f$, then from $A'_f - s$, 
and lastly from $B'_f - t$, we use the following lemma.
\begin{lemma}\label{lem: order-aux}
Fix an infeasible flow $f$.
If there is no path from $A_f'$ to $B_f'$ with positive capacity
in $G_f$, 
then sending flow from $A_f'$ to $s$ 
(or from $t$ to $B_f'$) to form the flow $h$
will not result in a path from $A_h'$ to $B_h'$ with positive capacity in $G_h$.
\end{lemma}
\begin{proof}
Assume for sake of deriving a contradiction
that (without loss of generality) sending flow from
some $u \in A_f'$ to $s$ formed flow $h$, for
which there exists a path from $A_f'$ to $B_f'$ with positive capacity in $G_h$.
Let $p_1 = (u_1, \ldots,s)$ be the path along which flow was sent in $G_f$, 
and let $p_2 = (u_2, \ldots,v)$ be the path with positive capacity in $G_h$. 
Note that $u_1,u_2 \in A_f'$ and $v \in B_f'$ by the assumptions.
Since $p_2$ went from having 0 capacity in $G_f$ to positive capacity in $G_h$, at least one edge $e=(a,b)$ of $p_2$ has $\cev{e}=(b,a)$ in $p_1$.
If there are multiple such edges take $e$ to be the last such edge in $p_1.$
Let $p_1' = (u_1,\ldots,a)$ be the truncation of $p_1$ at $a$
and let $p_2' = (a,\ldots,v)$ be the suffix of $p_2$ that begins at $a$.
Then let $p'$ be the concatenation of $p_1'$ and $p_2'$. 
Note that $p_1'$ and $p_2'$ both have positive capacity in $G_f$. 
Thus $p'$ is a path in $G_f$ with positive capacity from $u$ to $v$, 
which is a contradiction.
\end{proof}

\subsection{Warm-start algorithm analysis}
\label{sec:ws-analysis}
In this section we analyze how Algorithm \ref{alg:EK-warmstart} works.

\noindent \textbf{Validity of algorithm}\quad
We prove that a projection path must exist in the main \textsf{WHILE} loop.
\begin{lemma}
\label{lem:proj-path-exists}
Given infeasible flow $f$, $\forall u \in A'_f$, $\exists v \in B_f$ such that there is a projection path from $u$ to $v$; $\forall v \in B'_f$, $\exists u \in A_f$ such that there is a projection path from $u$ to $v$.
\end{lemma}

This lemma results from the following observation that links the summation of excess/deficit to the difference between in-flows and out-flows for any fixed set of nodes.
\begin{proposition}
\label{prop: ex-equals-def}
Let $f$ be a flow satisfying the capacity constraints of a network $G$.
Then for any $S \subseteq V$, the difference between the total deficits in $S$ and the total excesses in $S$ is exactly the difference between the total flow out of $S$ and into $S$. Formally,
\[
\sum_{u \in S} \textsf{def}_f(u) -\sum_{u \in S} \textsf{ex}_f(u)
=\sum_{u \in S}\sum_{\underset{v \not \in S}{e=(u,v): }} f_{e} -
\sum_{u \in S}\sum_{\underset{v \not \in S}{e=(v,u):}} f_{e}.
\]
\end{proposition}
\begin{proof} In 
$\sum_{u \in S} \textsf{def}_f(u)-\sum_{u \in S} \textsf{ex}_f(u)$, 
the edges with flow within $S$ are counted with a positive and negative sign, 
but the edges carrying flow into $S$ or out of $S$ are counted 
once by the excess and once by the deficit, respectively.
\end{proof}

\begin{proof}[Proof of Lemma \ref{lem:proj-path-exists}]
We prove one direction by contradiction. The proof for the other direction is similar. For any $u \in A'_f$, assume no such path exists. In Proposition \ref{prop: ex-equals-def} take $S$
to be the set of all vertices reachable from $u$ in $G_f$. None of the nodes in $S$ can have positive deficit, so the LHS of Proposition \ref{prop: ex-equals-def} must be negative. On the other hand, $S$ must have 0 flow incoming to it, otherwise there is an edge pointing from $S$ to $V \setminus S$ in $G_f$, producing a vertex in $V \setminus S$ reachable from $u$ and contradicting the maximality of $S$.
Therefore, the RHS in Proposition \ref{prop: ex-equals-def} is non-negative, contradicting the equation.
\end{proof}

Note each iteration decreases the total amount of excess and deficit in the system,
$\textsf{ex}_f +\textsf{def}_f$, by at least one, so the \textsf{WHILE} loop terminates after restoring flow conservation, giving rise to a feasible flow. Then the vanilla Ford-Fulkerson takes over until an optimal solution is found.

\noindent \textbf{Running-time analysis}\quad
Here we work towards proving Theorem \ref{thm: run-time-flow-0}. 
We show the running time is tied to the quality of prediction, $\eta(\widehat{f}) = \min_{f^* \in \mathcal{F}^*}||\widehat{f} - f^*||_1$. We first bound the times the path-finding subroutine is called. For projection paths, the total excess and deficit could increase by at most $\sum_e \max\{\widehat{f}_e-c_e,0\}$ when Algorithm \ref{alg:EK-warmstart} rounds down the flow where it exceeds capacity. Thus it takes at most $(\textsf{ex}_{\widehat{f}}+\textsf{def}_{\widehat{f}})$ projection paths to restore feasibility. For augmenting paths, the difference in flow value, $\nu(\widehat{f}) - \nu(f^*)$, could decrease by at most $\sum_e \max\{\widehat{f}_e-c_e,0\}$ during the round-down and another $\textsf{ex}_{\widehat{f}}+\textsf{def}_{\widehat{f}}$ during feasibility projection, thus the total number is at most the summation of the three. Each path-finding takes $O(|E|)$ time. This combined with the next lemma proves Theorem \ref{thm: run-time-flow-0}.
\begin{lemma}\label{lem: ub-rt-terms}
$\nu(f^*)-\nu(f),$ $\textsf{ex}_{\widehat{f}}+\textsf{def}_{\widehat{f}}$, 
and $\sum_e \max\{\widehat{f}_e-c_e,0\}$ are upper bounded by $\eta(\widehat{f})$.
\end{lemma}

\begin{proof}
First, we see that $|\nu(f) - \nu(f^*)|$ and $|\nu(\widehat{f})-\nu(f^*)|$ can only differ by
however much value was lost and however much excess and deficit was gained in projecting $\widehat{f}$ to
the feasible flow $f.$
Therefore, we can upper bound $|\nu(f) - \nu(f^*)| $ by
\begin{align*}
|\nu(f) - \nu(f^*)| \leq |\nu(\widehat{f})-\nu(f^*)| &+ \sum_e \max\{\widehat{f}_e-c_e,0\}
+\textsf{ex}_{\widehat{f}}+\textsf{def}_{\widehat{f}}.
\end{align*}
Then, we can further upper bound $|\nu(\widehat{f})-\nu(f^*)|$ by rewriting
the difference between the values of the flows
$$|\nu(f^*) - \nu(\widehat{f})| = \bigg | \sum_{e = (s,v)}f^*_e-\sum_{e = (s,v)}\widehat{f}_e  \bigg | \leq \eta(\widehat{f}).$$
The next term to bound in terms of the $\ell_1$ error is $\sum_e \max\{\widehat{f}_e-c_e,0\}$,
though it is straight forward to see
$$\sum_e \max\{\widehat{f}_e-c_e,0\} \leq \sum_e \max\{\widehat{f}_e-f^*_e,0\} \leq \eta(\widehat{f}).$$
Lastly, we see that the excess/ deficit of any node $v \in V \setminus \{s,t\}$
can be charged to the difference between $\widehat{f}_e$ and $f^*_e$ for $e$ adjacent to $v$, 
as any $f^* \in \mathcal{F}^*$ has excess/ deficit 0 on all non-source and sink nodes.
Therefore, $\textsf{ex}_{\widehat{f}}+ \textsf{def}_{\widehat{f}} \leq  \eta(\widehat{f}).$
\end{proof}

\begin{proof}[Proof of Theorem \ref{thm: run-time-flow-0}]

Given a flow $\widehat{f}$ that does not satisfy capacity constraints, 
Algorithm \ref{alg:EK-warmstart} simply updates the edges
$E' \subseteq E$ that violate capacity $\widehat{f}_e>c_e$
to $\widehat{f}_e \leftarrow c_e$ for $e \in E'$.
This can be done in time $O(|E'|)$. 
Further, rounding down the flow on these edges 
changes the value of the flow and the sum of the excess and deficit by at most
$\sum_e \max\{\widehat{f}_e-c_e,0\}$.

In Lemma \ref{lem:proj-path-exists}, we showed 
that given $\widehat{f}$ that satisfies capacity constraints, 
Algorithm \ref{alg:EK-warmstart} will find an optimal $f^*$.
Next, we analyze the run-time of Algorithm \ref{alg:EK-warmstart}.
Each iteration of the main while loop in Algorithm \ref{alg:EK-warmstart} costs time $O(|E|)$.
Further, the number of iterations in the main while loop in Algorithm \ref{alg:EK-warmstart} is at most
$\textsf{ex}_{\widehat{f}}+\textsf{def}_{\widehat{f}}$.
Let $f$ be the feasible flow obtained by Algorithm \ref{alg:EK-warmstart} at the 
end of the main while loop.
The run-time to produce flow $f$ is
$O(|E|(\textsf{ex}_{\widehat{f}}+\textsf{def}_{\widehat{f}}))$. 

At most $|\nu(f^*)-\nu(f)|$ iterations of any Ford-Fulkerson procedure are needed to arrive at the optimal flow value $\nu(f^*)$ from $f$ by Proposition \ref{prop: feasibile-ws}. 
Each iteration of Ford-Fulkerson also costs $O(|E|)$.

Therefore, the run-time of Algorithm \ref{alg:EK-warmstart} given a prediction which satisfies capacity constraints is at most
$$O (|E| \cdot ( |\nu(f) - \nu(f^*)|+ \textsf{ex}_{\widehat{f}}+\textsf{def}_{\widehat{f}}  )).$$
Combining this with the loss from projecting to satisfy capacity constraints, 
the full run-time of Algorithm \ref{alg:EK-warmstart} is 
\begin{align*}
O\big (|E| \cdot \big ( &\sum_e \max\{\widehat{f}_e-c_e,0\}+|\nu(f) - \nu(f^*)|
+\textsf{ex}_{\widehat{f}}+\textsf{def}_{\widehat{f}} \big )\big).
\end{align*}
We showed that all of the terms multiplying $|E|$ in the above can be bounded by $O(\eta(\widehat{f}))$ (Lemma \ref{lem: ub-rt-terms}).

It is straight-forward to justify the worst-case run-time of Algorithm \ref{alg:EK-warmstart} is bounded by $O(T)$. During the feasibility projection step an auxiliary graph is constructed only three times (recall that from Lemma \ref{lem: order-aux}, we can send flow on $A_f' - B_f'$ paths, then $A_f' - s$ paths, then $t - B_f'$ paths), 
each time with $|V| + 2$ vertices and $O(|E|)$ edges. 
Thus running the chosen Ford-Fulkerson implementation on these graphs takes time $O(T)$. Recall that the optimization step also takes time $O(T)$, since running Ford-Fulkerson starting with a feasible flow is equivalent to running it from scratch on the residual graph as a new input. Thus the total running time is $O(T)$.
\end{proof}

\subsection{PAC-learning Flows}
\label{sec:PAC-learning}

Here, we show theoretically that high quality flows are learnable, thus
giving evidence that flows can be learned for input to Algorithm \ref{alg:EK-warmstart}.
We show that
given a distribution over capacity vectors for a network, 
one can learn a predicted flow from samples that is the best approximation. 

Consider a fixed network $G$ with integral edge capacities $c$. 
An instance is a network $G^i$ on the same vertex and edge set as $G$, 
but the capacity vector is $c^i$, 
where every edge $e$ in $G^i$ must satisfy $c^i_e \in [0,c_{e}] \cap \mathbb{Z}$.
Let $\mathcal{D}$ be an unknown distribution over such instances.
Since an instance is exactly characterized by its new capacity vector, 
we notationally write this as sampling a capacity vector  $c^i \sim \mathcal{D}$.

Suppose we sample instances $c^1,\ldots,c^s$ from $\mathcal{D}$. 
Let $\mathcal{F}$ be the set of all integral flows on $G$ that satisfy the capacities in $c$,
noting that flows in $\mathcal{F}$ do not have to satisfy flow conservation.
Technically, a network $G$ might have several optimal solutions. Here we make the following assumption.

\begin{assumption}\label{asmp: unique}
For a network $G$, there is a uniquely associated, computable optimal flow.\footnote{Such assumptions are standard for the PAC-learning results in learning-augmented based run-time improvements, even if not explicitly stated. See, for instance \cite{SakaueO22}.}
\end{assumption}

Given samples $c^1,\ldots,c^s$, 
we can compute the uniquely associated optimal flows on samples $f^*(c^1)$,$\ldots$, $f^*(c^s)$.
Let $\widehat{f}$ denote a predicted flow.
When our goal is to warm-start Ford-Fulkerson, 
we choose the predicted flow to be that in $\mathcal{F}$ which minimizes the empirical risk 
$\widehat{f}= \text{argmin}_{f \in \mathcal{F}}\frac1s \sum_{j=1}^s ||\widehat{f} - f^*(c^j)||_1.$, 
which given Assumption \ref{asmp: unique} can be efficiently computed, as in Lemma \ref{lem: same-opt}. 

\begin{lemma}
\label{lem: same-opt}
    One can find a flow $\widehat{f} \in \mathcal{F}$ minimizing
    $\frac1s \sum_{j=1}^s ||\widehat{f} - f^*(c^j)||_1$
    from independent samples 
    $c^1,\ldots,c^s \sim \mathcal{D}$ in polynomial time by taking  
    $\widehat{f}_e = \textsf{median}(f^*(c^1)_e,\ldots,f^*(c^s)_e)$ 
    for all $e \in E$.
\end{lemma}

\begin{proof}
We would like to find $\widehat{f} \in \mathcal{F}$
that minimizes $\frac1s \sum_{j=1}^s ||\widehat{f} - f^*(c^j)||_1.$
Since we do not require flow conservation, the minimization can occur over each edge independently, 
where $\widehat{f}_e$ will be in $[0,c_{e}] \cap \mathbb{Z}$, 
i.e. it suffices to minimize $\frac1s \sum_{j=1}^s |\widehat{f}_e - f^*(c^j)_e|$ for each $e \in E$.
The function $\frac1s \sum_{j=1}^s |\widehat{f}_e - f^*(c^j)_e|$
is continuous and piece-wise linear in $\widehat{f}_e$, 
where the slope changes at the points $\{f^*(c^j)_e\}_j$. 
It is well-known that the minimum of this function in $[0,c_{e}]$ is $\textsf{median}(f^*(c^1)_e,\ldots,f^*(c^s)_e)$.
\end{proof}

We will now state our PAC-learning result.
The proof of this theorem follows that of \citet{DinitzILMV21}.

\begin{thm}
\label{thm: learning-PAC}
Let $c^1,\ldots,c^s$ be sampled i.i.d. from $\mathcal{D}$ and let $\widehat{f}= \text{argmin}_{f \in \mathcal{F}_0}\frac1s \sum_{j=1}^s ||f - f^*(c_j)||_1.$ For $$s=\Omega((\max_{e}c_{e}^2 \cdot  m^2) (\log m + \log(1/p)))$$ and $\widetilde{f} = \text{argmin}_{f \in \mathcal{F}_0} \hspace{1mm} \mathbb{E}_{c^i \sim \mathcal{D}}||f-f^*(c^i)||_1$, then with probability at least $1-p$, $\widehat{f}$ satisfies
$$
\mathbb{E}_{c^i \sim \mathcal{D}} ||\widehat{f}-f^*(c^i)||_1 \leq \mathbb{E}_{c^i \sim \mathcal{D}} ||\widetilde{f}-f^*(c^i)||_1 + O(1).
$$
\end{thm}

In the proof of Theorem \ref{thm: learning-PAC}, 
we will use some well-known results regarding the pseudo-dimension of a class of functions.

The VC dimension is a quantity that captures the complexity of a family of binary functions, and
the pseudo-dimension is the analog of this for real-valued functions
Specifically, the \textbf{pseudo-dimension} of a family of real-valued functions $\mathcal{H}$
is the largest sized subset shattered by $\mathcal{H}$. 
A subset $S = \{x_1,\ldots,x_s\}$ of $X$ is \textbf{shattered} by $\mathcal{H}$ if there exists real-valued witnesses $r_1,\cdots,r_s$ such that for each of the $2^s$ subsets $T$ of $S$,
there exists a function $h \in \mathcal{H}$ with $h(x_i) \leq r_i$ if and only if $i \in T$.

The following theorem relates the convergence of the sample mean of some $h \in \mathcal{H}$ to its expectation, 
and this relation depends on the pseudo-dimension.
\begin{thm}[Uniform convergence]
\label{thm:uniform_convergence}
Let $\mathcal{H}$ be a class of functions with domain $X$ and range in $[0,H]$. Let $d_{\mathcal{H}}$ be the pseudo-dimension of $\mathcal{H}$. For every distribution $\mathcal{D}$ over $X$, every $\epsilon >0$, and every $\delta \in (0,1]$, if 
$$ 
s \geq c (H/\epsilon)^2 (d_{\mathcal{H}} + \ln (1/\delta))
$$
for some constant $c$, then with prob at least $1-\delta$ over $s$ samples $x_1\ldots,x_s \in \mathcal{D}$,
$$
\left | \left (  \frac{1}{s} \sum_{i=1}^s h(x_i)\right ) - \mathbb{E}_{x \sim\mathcal{D}} [h(x)]\right | < \epsilon.
$$
\end{thm}

Equipped with Theorem \ref{thm:uniform_convergence}, 
we are ready to prove our PAC-learning result.

\begin{proof}[Proof of Theorem \ref{thm: learning-PAC}]
We will construct a class of functions that contains the loss functions of the flow $f$
given capacity constraints $c^i$. 
Then, we will apply Theorem \ref{thm:uniform_convergence} to this class of functions.

For every integral flow $f \in \mathbb{R}^{|E|}$ that satisfies the capacity vector $c^i$, 
let the function $g_f(c^i) = || f^*(c^i) - f ||_1$ be the loss function of $f$ on $c^i$.
Then let $\mathcal{H} = \{g_f \mid f \in \mathbb{R}^m\}$ 
be the family of all of these loss functions.

We saw in Lemma \ref{lem: same-opt} how to efficiently compute the empirical risk minimizer. 
Also, the upper bound of the range of the loss functions, 
i.e. $H$ in the statement of Theorem \ref{thm:uniform_convergence}, 
is at most $m \cdot \max_e c_{e}$.
To prove our lemma, it remains to bound the pseudo-dimension of $\mathcal{H}$.

We will upper bound the pseudo-dimension of $\mathcal{H}$ by showing it is no more than the pseudo-dimension of another class of functions, $\mathcal{H}_m$, whose pseudo-dimension is already known.
Let $\mathcal{H}_m= \{h_y \mid y \in \mathbb{R}^m\}$ for $h_y(x) = ||y-x||_1$.
The following result appears as Theorem 19 in \cite{DinitzILMV21},
and the reader may refer to that paper for its proof. 
\begin{lemma}
The pseudo-dimension of $\mathcal{H}_m$ is at most $O(m \log m).$
\end{lemma}

Now all that remains is to prove the following lemma, 
relating the pseudo-dimensions of the two classes.
\begin{lemma}
If the pseudo-dimension of $\mathcal{H}_m$ is at most $d$, 
then the pseudo-dimension of $\mathcal{H}$ is at most $d$.
\end{lemma}

\begin{proof}
Let $S = \{c^1,\ldots,c^d\}$ be a set that is shattered by $\mathcal{H}$. Let $r_1,\ldots,r_d \in \mathbb{R}$ be the witnesses such that for all $S' \subseteq [d]$, there exists some $g_{f^{S'}} \in \mathcal{H}$ with $g_{f^{S'}}(c^j)=||f^*(c^j)-f^{S'}||_1 \leq r_j$ if and only if $j \in S'$. 

We will construct a set $\widetilde{S}$ of size $d$ from $S$ that is shattered by $\mathcal{H}_m$. Let $\widetilde{S} = \{f^*(c^1),\ldots,f^*(c^d)\}$ and fix some $S' \subseteq [d]$. Then 
$h_{f^{S'}}(f^*(c^j))=||f^{S'}-f^*(c^j) ||_1\leq r_j$ if and only if $j \in S'$.
\end{proof}

Plugging $H \leq m \cdot \max_e c_{e}$ and $d_{\mathcal{H}} \leq O(m \log m)$ into Theorem \ref{thm:uniform_convergence}, we see that 
it suffices to take 
$$ 
s \geq \Omega( (\max_e c_{e}\cdot m /\epsilon)^2 (m \log m + \ln (1/\delta))).
$$
\end{proof}

\section{Faster Flows via Shorter Projection Paths}
\label{sec: instance-robust}

Here, we show that Algorithm \ref{alg:EK-warmstart} 
can be even faster for a certain class of networks.  
Intuitively, the additional speed-up is obtained due to finding shorter projection paths.  
We will then explain in Section \ref{sec: empirical} how simple 
image segmentation networks fit into this class,
so this theory explains the speed-up we see on the image segmentation instances.

Let $G=(V,E)$ be a directed graph, with $s,t \in V$.
Suppose $V \setminus \{s,t\}$ forms a two-dimensional grid.
Further for $u,v \in V \setminus \{s,t\}$,
if $e = (u,v) \in E$ then the reverse direction edge $\cev{e}=(v,u) \in E$.
We consider a pair of networks $G^1,G^2$ on $G$.
The only difference in these networks is their capacity vectors,
though we assume they have capacity vectors $c^1,c^2 \in \{1,M\}^m$ 
for some large integer $M$, 
and we assume that all edges incident to $s$ or $t$ have capacity $M$.

For $\ell \in [2]$, let $E_\ell = \{e \in E \mid c^\ell_e =1 \}$.
We call $G^\ell$ \textbf{separable}
if the vertices in $V \setminus \{s,t\}$ can be partitioned into subsets
$V_\ell$ and $W_\ell = V \setminus (\{s,t\} \cup V_\ell )$,
such that there is some $x \in V_1 \cap V_2$ with $(x,t) \in E$ 
and $y \in W_1 \cap W_2$ with $(s,y) \in E$,
and for all $e=(u,v) \in E_\ell$, 
$e$ has one endpoint in $V_{\ell}$ and the other in $W_\ell$.
We say the transition between $G^1$ and $G^{2}$ is \textbf{d-local} if
for all pairs of distinct nodes $u,v \in (V_{2} \setminus V_1) \cup  (W_{2} \setminus W_1)$,
their distance, i.e., the length of the shortest path between them, is at most $d$.
Here, $d$ controls the length of the projection paths.

While we require  additional assumptions for our theoretical results,
our empirical results  in Section \ref{sec: empirical},
show that these assumptions are sufficient but not necessary for Algorithm \ref{alg:EK-warmstart}
to take advantage of short projection paths.

For the proof of Theorem \ref{thm: instance-robust-0},
the feasibility projection of $\widehat{f}$  has several steps.
First, Algorithm \ref{alg:EK-warmstart} can choose paths along which to 
route flow so that either the excess and deficit are fixed,
or $V_2$ only contains nodes with positive excess and 
$W_2$ only contains nodes with positive deficit
or vice versa.
We argue that this fixes the total excess and deficit 
by at most $\eta(\widehat{f})$ with paths of length $O(d)$.
Then, it remains to fix the excess/ deficit in 
$V_2$ and $W_2$ using $s$ and $t$.
These paths have unbounded length, i.e. length $O(|E|)$, 
and we can argue the change in flow value is at most $||E_2| - |E_1|| \leq O(d^2)$, 
where the upper bound comes from he definition of $d$-local.

\begin{thm}[Restates Theorem \ref{thm: instance-robust-0}]
\label{thm:instance-robustness-formal}
Fix separable networks $G^1$ and $G^2$, 
where the transition between 
them is $d$-local.
For $\widehat{f}$ an optimal flow on $G^1$, 
the run-time of Algorithm \ref{alg:EK-warmstart} seeded with $\widehat{f}$ on 
$G^2$ to find an optimal $f^*$ is
$O(d^2 \cdot |E| +d \cdot \eta(\widehat{f}))$.
\end{thm}

\begin{proof}
For the network $G^{2}$ with capacity constraints $c^{2}$,
project $\widehat{f}$ to satisfy $c^{2}$ as in Algorithm \ref{alg:EK-warmstart},
and let the resulting flow be $f$.
Note that a node can only be in $A_f' \cup B_f'$ if it is incident to an edge
whose flow was rounded down in the projection from $\widehat{f}$ to $f$.
An edge $(u,v)$ has flow rounded down in the projection 
only if it went from being an edge with both $u$ and $v$ either in $V_1$ or outside of $V_{1}$
to being a boundary edge, 
i.e. one of $u$ or $v$ is in $V_{2}$ and one is in $W_{2}$.

To project $f$ to a flow that satisfies flow conservation,
Algorithm \ref{alg:EK-warmstart} can choose paths to 
first route flow from nodes with positive excess in $V_{2}$ 
to nodes with positive deficit in $V_{2}$,
and route flow from nodes with positive excess in $W_{2}$ 
to nodes with positive deficit in $W_{2}$. 
By our assumption that the networks are $d-$local,
excess and deficit within $V_{2}$ or in $W_{2}$ have distance 
at most $d$.
Since the capacity of the non-boundary edges is $M$, 
Algorithm \ref{alg:EK-warmstart} can route flow until there is only excess or only deficit contained
within $V_{2}$ and within $W_{2}$ with these projection paths. 
Additionally, if $V_{2}$ contains a node with positive deficit and $W_{2}$ 
contains a node with positive excess and
if there are any edges with positive flow going from $V_{2}$ to $W_{2}$,
one can send flow on the reverse edge and remove that excess/ deficit 
with a projection path of length at most $2d$.

Further, Algorithm \ref{alg:EK-warmstart} can choose paths to perform this flow routing so excess and deficit are symmetric across the boundary.
Specifically, for a node $u \in V_{2}$, one can route flow so
$\textsf{ex}_f(u) = \sum_{v: (u,v) \in E_{2}}\textsf{def}_f(v)$, 
and this holds for the deficits, and analogously for the nodes of $W_{2}$.
Note that this re-routing is possible since $M$ is sufficiently large.
From the proof of Theorem \ref{thm: run-time-flow-0}, we know that 
$\textsf{ex}_f+\textsf{def}_f \leq \eta(\widehat{f})$, 
so the run-time of re-routing this excess and deficit is at most $O(d \cdot  \eta(\widehat{f}))$.

We will show any remaining deficit/ excess can be handled by paths using $s$ and $t$.
First, the max flow on $G^{1}$ has value $|E_1|$,
since the edges crossing from $W_1$ into $V_1$ form a cut.
On the other hand, for large $M$ and by the existence of $x \in V_1 \cap V_2$ 
and $y \in W_1 \cap W_2$,
there exists feasible flows with every edge $(u,v)$ 
with $u \in W_1$ and $v \in V_{1}$ having flow 1 incoming to $V_1$.
Second, our re-routing procedure made 
the excess within $V_{2}$ symmetric across the boundary 
to the deficit in $W_{2}$, 
so there is either positive excess in $V_{2}$ and positive deficit in $W_{2}$,
or vice versa.

We  use these two observations.
Suppose that after  re-routing, there is deficit inside of $V_{2}$
and excess outside of it.
As a consequence of our routing, 
we can assume all edges with positive flow in $E_{2}$ 
are directed from $W_{2}$ to $V_{2}$.
For every node $u$ with positive excess in $W_{2}$, 
take the excess of $u$ and send it to $s$, 
which is possible from the conditions for being separable. 
Similarly, for every node $u$ with positive deficit inside of $V_{2}$,
find paths from $t$ to $u$ and send $\textsf{def}_f(u)$ from $t$ to $u$, 
and this is again possible by the separable condition.
The resulting flow is feasible.
Further, there are no $s-t$ paths since all edges in $E_{2}$ going into $V_2$ are saturated
and form a cut, so the flow is optimal.
So re-routing this flow using $s$ and $t$ takes time 
$O(|E| (|E_{1}|-|E_{2}|))$.

When after  re-routing 
there is excess inside of $V_{2}$ and deficit inside of $W_{2}$,
the proof is similar.
The edges with no excess/ deficit incident to them have flow 1 going into $V_{2}$.
Since the deficit is outside of $V_{2}$,
the boundary edges crossing from $V_{2}$ into $W_{2}$ 
have flow 1 going out of $V_{2}$. 
Send flow from $s$ to the nodes with positive deficit inside of $W_{2}$,
and send the excess inside of $V_{2}$ to $t$, 
which is possible by the conditions for being separable. 
The resulting flow is feasible, though perhaps not optimal, as one may need to saturate 
new boundary edges. 
The run-time is still $O(|E| (|E_{2}|-|E_{1}|))$.

By the definition of $d-$local, 
 $\left ||E_{2}|-|E_{1}| \right | \leq O(d^2)$.
Therefore, the run-time of Algorithm \ref{alg:EK-warmstart} 
on these locally-changed instances is at most 
$O(d^2 \cdot |E|+ d \cdot \eta(\widehat{f})))$.
\end{proof}

\section{Empirical Results}
\label{sec: empirical}
In this section, we validate the theoretical results in Sections \ref{sec: FF} and \ref{sec: instance-robust}\footnote{The code is published at \url{https://github.com/wang-yuyan/warmstart-graphcut-algorithms-pulic}.}.We consider image segmentation, a core problem in computer vision that aims at separating an object from the background in a given image. 
The problem is re-formulated as a max-flow/ min-cut optimization problem in a line of work \citep{boykov2001interactive, boykov2004experimental, boykov2006graph} and solved with combinatorial graph-cut algorithms, including Ford-Fulkerson.

We do not attempt to provide state-of-the-art run-time results 
on image segmentation.
Our goal is to show that on real-world networks, 
warm-starting Ford-Fulkerson leads to big run-time improvements 
compared to cold-start Ford-Fulkerson.
We highlight the following:
\begin{compactitem}
    \item For both Edmonds-Karp and Dinic's implementation of Ford-Fulkerson, warm-start offers improved running time compared to starting the algorithm from scratch (referred to as a \textbf{cold-start}).
    \item As we increase the number of image pixels (i.e., its resolution), the size of the constructed graph increases and the savings in time becomes more significant.
    \item The feasibility projection step in Algorithm \ref{alg:EK-warmstart} has high performance. It returns a feasible flow that is only slightly sub-optimal, and it finds short paths to fix the excess/deficits in doing so. Both factors contribute to warm-start being way more efficient than cold-start.
\end{compactitem}

\noindent \textbf{Datasets and data prepossessing\quad}
We use four different image groups from the \emph{Pattern Recognition and Image Processing} dataset from the University of Freiburg\footnote{https://lmb.informatik.uni-freiburg.de/resources/datasets/}, named \textsc{Birdhouse}, \textsc{Head}, \textsc{Shoe}, and \textsc{Dog} respectively. The first three groups are from the dataset \emph{Image Sequences}\footnote{https://lmb.informatik.uni-freiburg.de/resources/datasets/sequences.en.html}, in the format of .jpg images, whereas \textsc{Dog}, from \emph{Stereo Ego-Motion Dataset}\footnote{https://lmb.informatik.uni-freiburg.de/resources/datasets/StereoEgomotion.en.html}, is a video which we converted to .jpg. 

Each image group contains a sequence of photos featuring the same object and background. The sequence may feature the object's motion relative to the background or changes in the camera's shooting angle. 
Any image is only slightly different from the previous one in the sequence, and this could potentially lead to minor differences in segmentation solutions. This justifies warm-starting with the optimal flow for the max flow problem found on the previous image.

\begin{table}[ht]
\centering
\caption{Image groups' descriptions}
\label{table:data_desc}
\begin{tabular}{r|llll}
\hline
    Image Group & Object, background & Original size & Cropped size \\ 
\hline
    \textsc{Birdhouse} & wood birdhouse, backyard & 1280, 720 & 600, 600 \\
    \textsc{Head} & a person's head, buildings & 1280, 720 & 600, 600 \\
    \textsc{Shoe} & a shoe, floor and other toys & 1280, 720 & 600, 600 \\
    \textsc{Dog} & Bernese Mountain dog, lawn & 1920, 1080 & 500, 500\\
\hline
\end{tabular}
\end{table}

We take 10 images from each group, 
cropped them to be $600 \times 600$ pixels with the object included, and gray-scaled them. 
Then we resize the images to generate image sequences of different sizes. 
See Table \ref{table:data_desc} for detailed information about the image groups, the featured object/backround, and the original and cropped sizes of each image.  
See Figures \ref{fig:original_images} and \ref{fig:cropped_images} for an image instance from each group.

\begin{figure}[ht]
        \centering
        \begin{subfigure}[b]{0.25\linewidth}
            \centering
            \includegraphics[width=\linewidth]{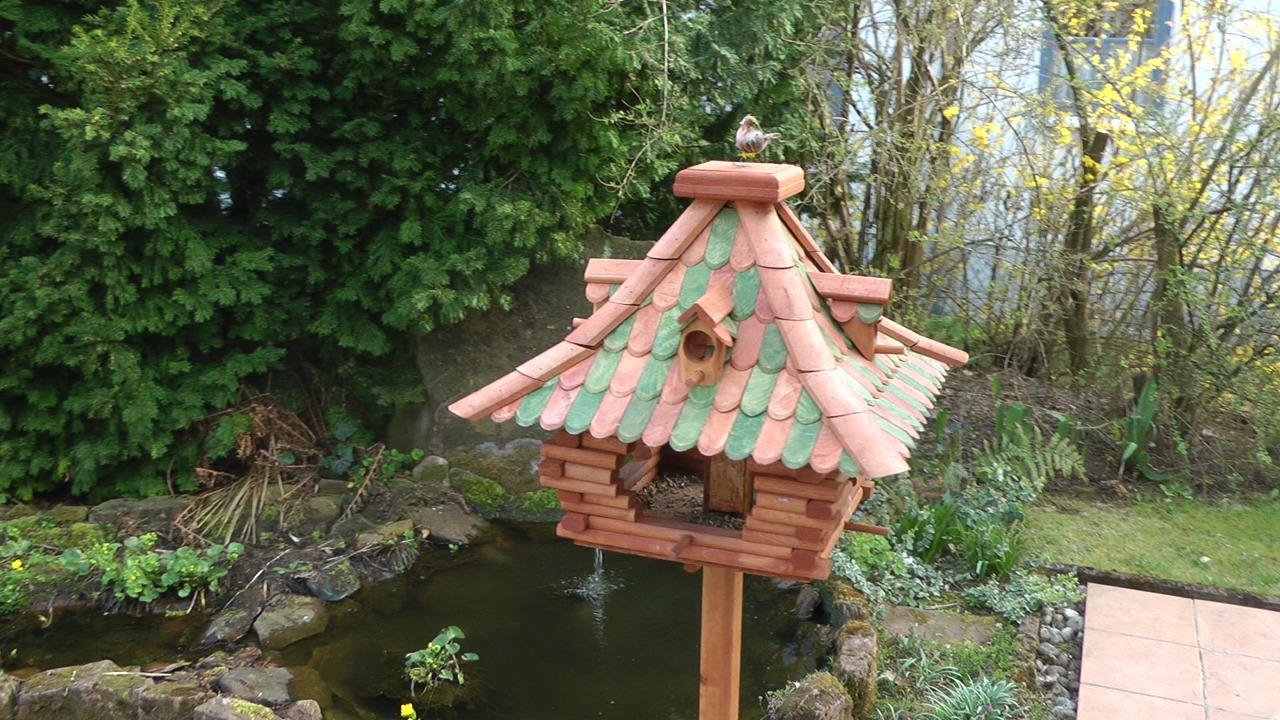}
            \caption[]%
            {{\small Birdhouse}}    
            \label{fig:birdhouse}
        \end{subfigure}
        \hspace{1cm}
        \begin{subfigure}[b]{0.25\linewidth}  
            \centering 
            \includegraphics[width=\linewidth]{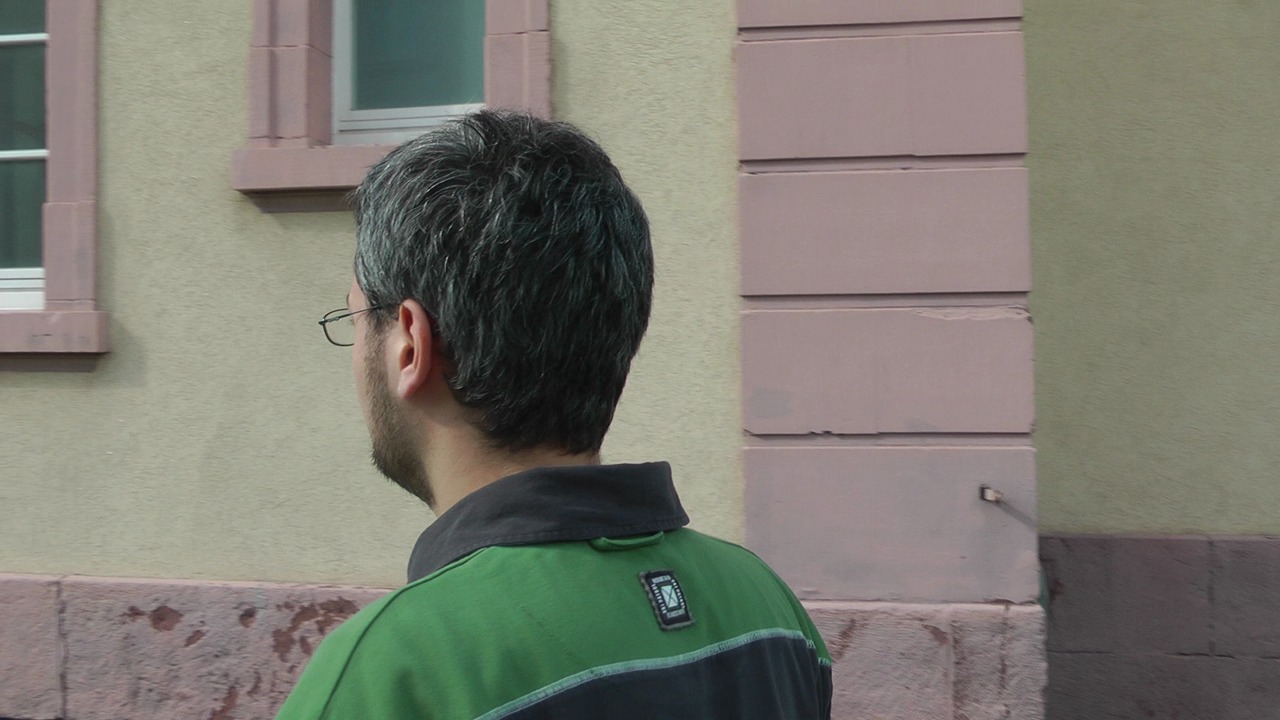}
            \caption[]%
            {{\small Head}}    
            \label{fig:head}
        \end{subfigure}
        \vskip\baselineskip
        \begin{subfigure}[b]{0.25\linewidth}   
            \centering 
            \includegraphics[width=\linewidth]{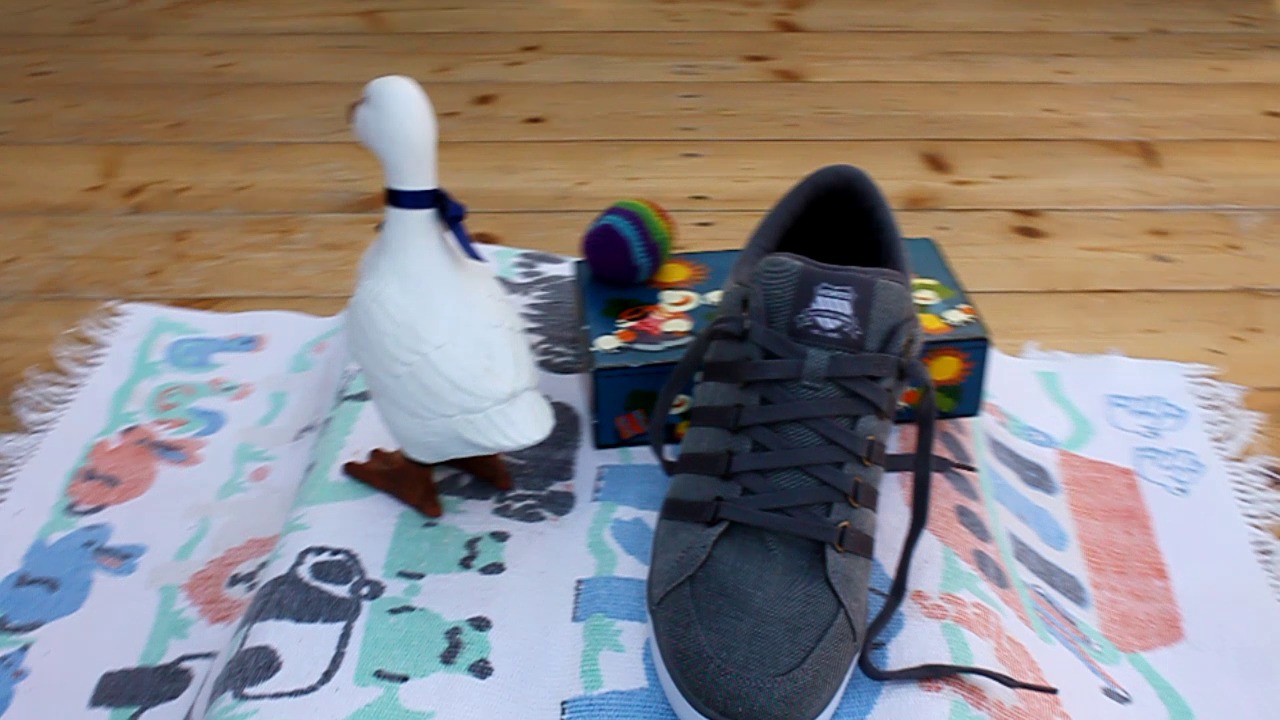}
            \caption[]%
            {{\small Shoe}}    
            \label{fig:shoe}
        \end{subfigure}
        \hspace{1 cm}
        \begin{subfigure}[b]{0.25\linewidth}   
            \centering 
            \includegraphics[width=\linewidth]{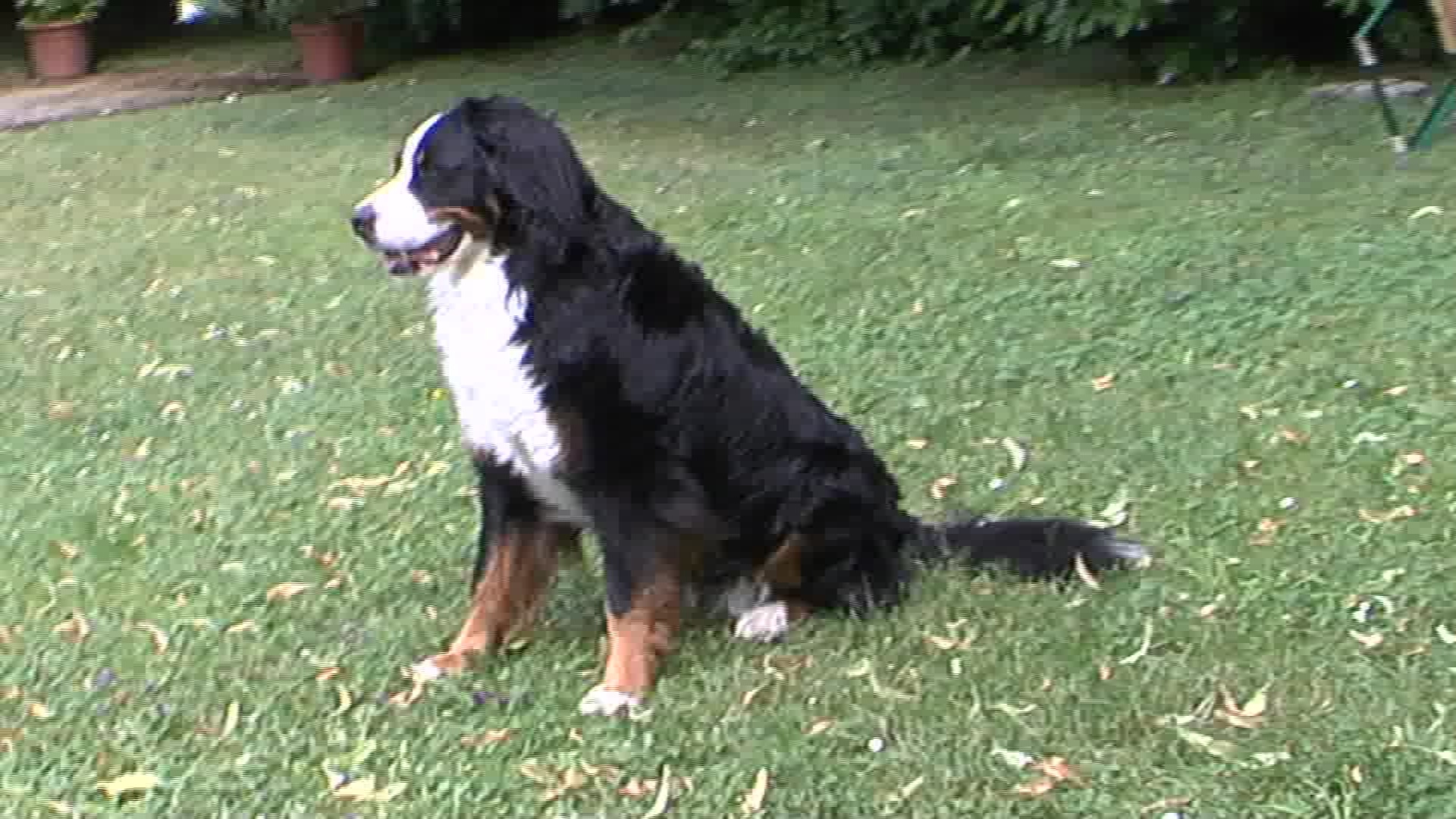}
            \caption[]%
            {{\small Dog}}    
            \label{fig:dog}
        \end{subfigure}
        \caption{\small Examples of original images in each group.} 
        \label{fig:original_images}
\end{figure}

\begin{figure}
        \centering
        \begin{subfigure}[b]{0.2\linewidth}
            \centering
            \includegraphics[width=\linewidth]{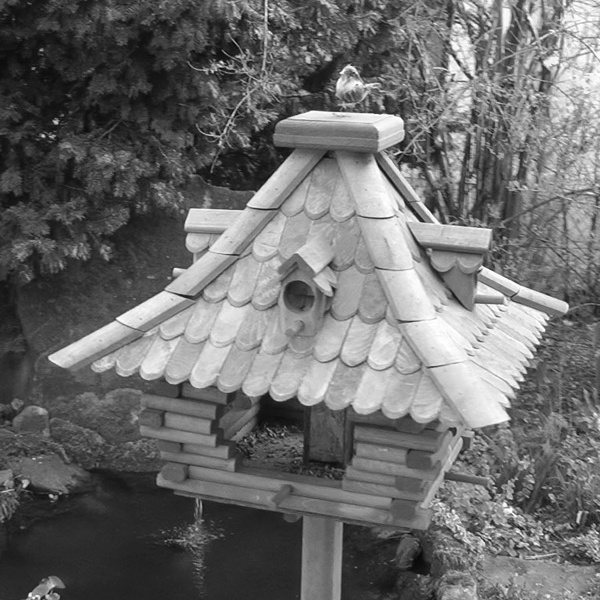}
            \caption[]%
            {{\small Birdhouse}}    
            \label{fig:birdhouse_c}
        \end{subfigure}
        \hfill
        \begin{subfigure}[b]{0.2\linewidth}  
            \centering 
            \includegraphics[width=\linewidth]{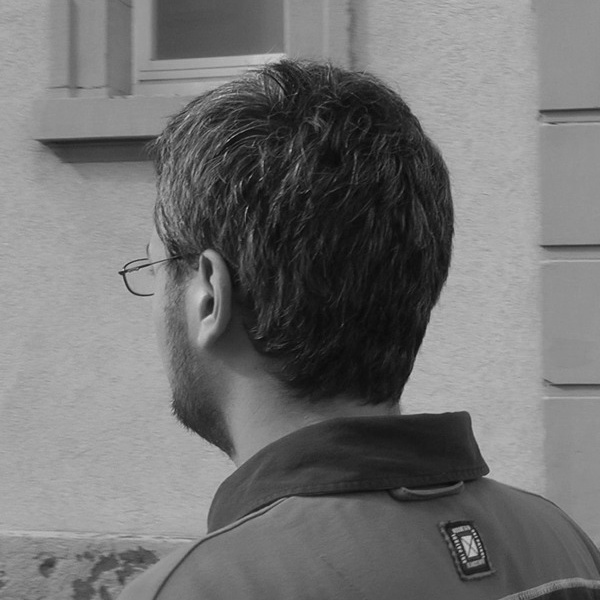}
            \caption[]%
            {{\small Head}}    
            \label{fig:head_c}
        \end{subfigure}
        \hfill
        \begin{subfigure}[b]{0.2\linewidth}   
            \centering 
            \includegraphics[width=\linewidth]{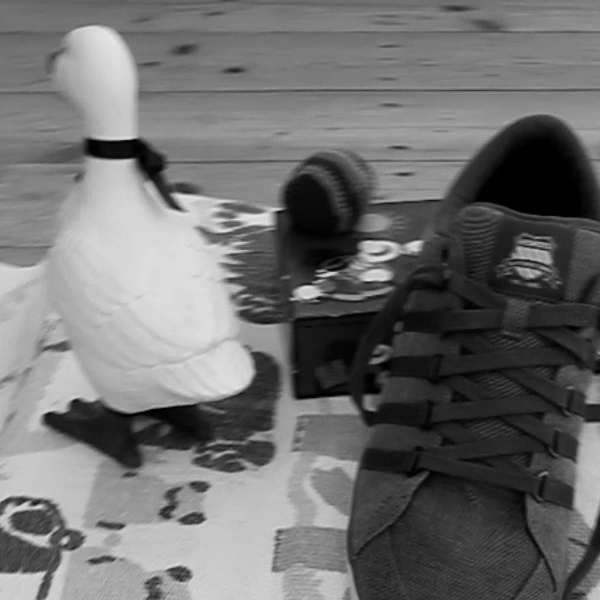}
            \caption[]%
            {{\small Shoe}}    
            \label{fig:shoe_c}
        \end{subfigure}
        \hfill
        \begin{subfigure}[b]{0.2\linewidth}   
            \centering 
            \includegraphics[width=\linewidth]{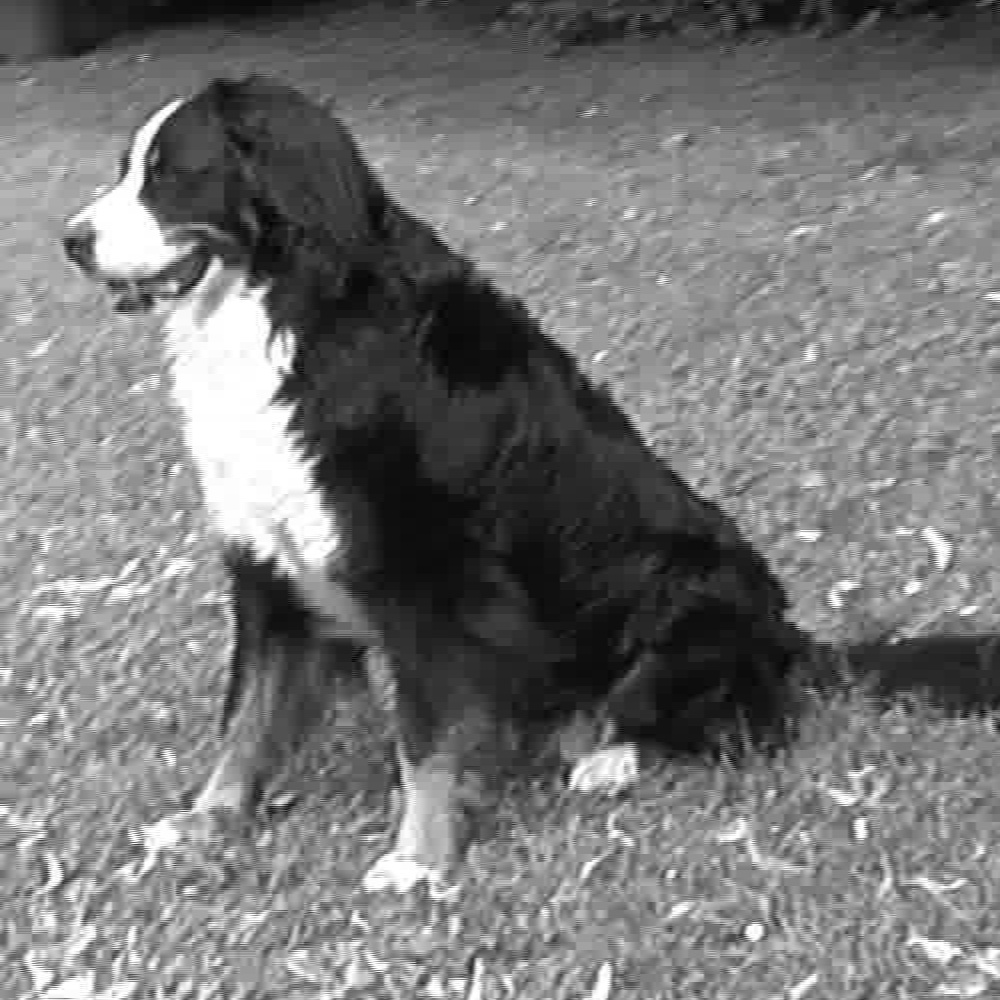}
            \caption[]%
            {{\small Dog}}    
            \label{fig:dog_c}
        \end{subfigure}
        \caption{\small Cropped, gray-scaled images in each group.} 
        \label{fig:cropped_images}
\end{figure}

\noindent \textbf{Graph construction}\quad
Following the practice in \citet{boykov2006graph}, we briefly describe how to formulate image segmentation as a max-flow/min-cut problem and how to write the boundary-based objective function. Our input is an image with pixel set $V$, along with two sets of \textbf{seeds} $\mathcal{O}, \mathcal{B}$, which are pixels predetermined to be part of the object or background, respectively (often selected by human experts), to make the segmentation task easier. Let $I_v$ denote the \textbf{intensity} (or gray scale) of pixel $v$. For any two pixels $p,q$, separating them in the object/background segmentation solution induces a \textbf{penalty} of $\beta_{p, q}$. If $p,q$ are neighboring pixels, i.e. $p$ and $q$ are either in the same column and in adjacent rows 
or same row and adjacent columns, then $\beta_{p,q}=C \exp (-\frac{(I_p - I_q)^2}{2\sigma^2})$, where $C$ is a relatively big constant scalar, otherwise it is $0$. Thus $\beta_{p,q}$ gets bigger with stronger contrast between neighboring $p$ and $q$. 

For a given solution let $J$ denote the object pixels. The \textbf{boundary-based} objective function is the summation of the penalties over all pairs of pixels:
$\max_{J}  \sum_{p \in J, q \notin J}\beta_{p,q},$ for $J$ satisfying $\mathcal{O} \subseteq J, \mathcal{B} \subseteq V \setminus J  $.
Penalties are only imposed on the object boundary. The best segmentation minimizes the total penalty, thus maximizing the contrast between the object and background across the boundary, while satisfying the constraints imposed by seeds.

This is equivalent to solving the max-flow/min-cut problem on the following graph. Let the node set be all the pixels plus two terminal nodes: the object terminal $s$ (source) and the background terminal $t$ (sink). We add the following arcs:
(1) from $s$ to every node in $\mathcal{O}$, with a huge capacity $M$;
(2) from every node in $\mathcal{B}$ to $t$, again with capacity $M$;
(3) from every pair of node $p,q \in V$ (including the seeds), both arcs $(p, q)$ and $(q, p)$ with capacity $\beta_{p, q}$.
The value $M$ should ensure that these arcs never appear in the optimal cut. The flow goes from $s$ to $t$. For an $n \times n$ pixels image, the graph is sparse with $O(n^2)$ nodes and also $O(n^2)$ arcs.

\noindent \textbf{Link to theory}\quad 
For an image sequence, the constructed graphs are a generalization of the setting in Section \ref{sec: instance-robust}. The graphs form 2-dimensional grids and share the same network structure, the only differences being the capacity vectors. In addition, Section \ref{sec: instance-robust} makes other assumptions which also translate into properties of the images. The 1 or $M$ edge capacities assumption implies an extreme contrast between the gray scales of object and background pixels. The $d$-local assumption says that from one image to the next, the new object and background pixels are geographically close, implying only minor changes in the object's shape and location. Our image sequences do not strictly satisfy these properties. However, in all of our experiments the conclusions remain robust against moderate violations of the theoretical assumptions, showing that warm-starts can be beneficial in practice beyond current theoretical limits.

\noindent \textbf{Detailed experiment settings}\quad
Each image sequence has 10 images and they share the same set of seeds, 
so the constructed graphs have the same structure.  
See Figure \ref{fig:seed_example} for seeds for \textsc{Birdhouse}. 
Starting with the second image, we reuse the old max-flow solution on the previous one and pass the flow to Algorithm \ref{alg:EK-warmstart}. During the feasibility projection,
we pick a node and keep diminishing its excess/deficit by finding a projection path and sending flow down that path, until excess/deficit is 0. As in Section \ref{sec: instance-robust}, we prioritize projection paths excluding $s$ and $t$, since these modifications preserve the overall flow value, and we only send flow back to $s$ and from $t$ when no other paths exist.

\begin{figure}
        \centering
        \begin{subfigure}[b]{0.2\linewidth}
            \centering
            \includegraphics[width=\linewidth]{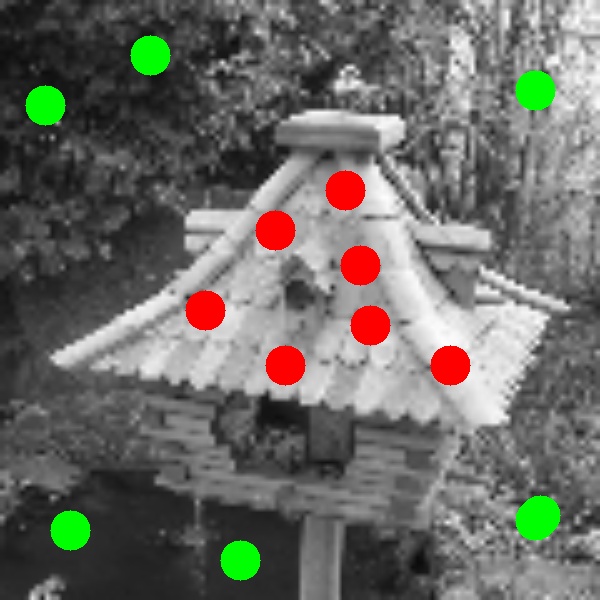}
            \caption[]%
            {{\small Image 1}}    
            \label{fig:bhseeds1}
        \end{subfigure}
       \hspace{1cm}
        \begin{subfigure}[b]{0.2\linewidth}  
            \centering 
            \includegraphics[width=\linewidth]{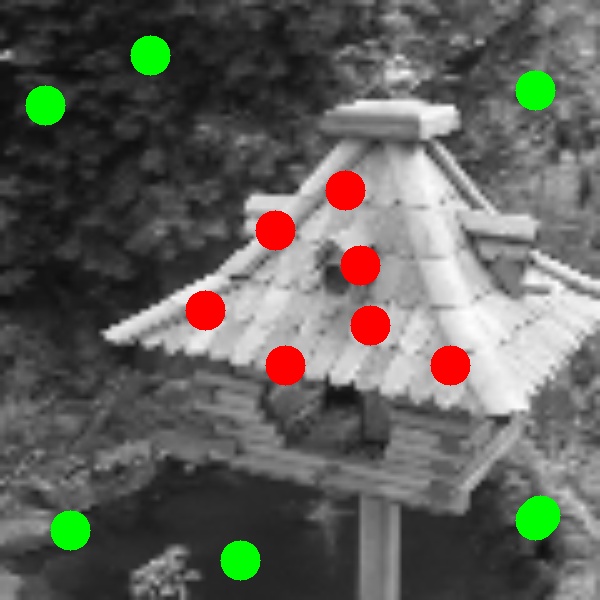}
            \caption[]%
            {{\small Image 5}}    
            \label{fig:bhseeds2}
        \end{subfigure}
        \hspace{1cm}
            \begin{subfigure}[b]{0.2\linewidth}  
            \centering 
            \includegraphics[width=\linewidth]{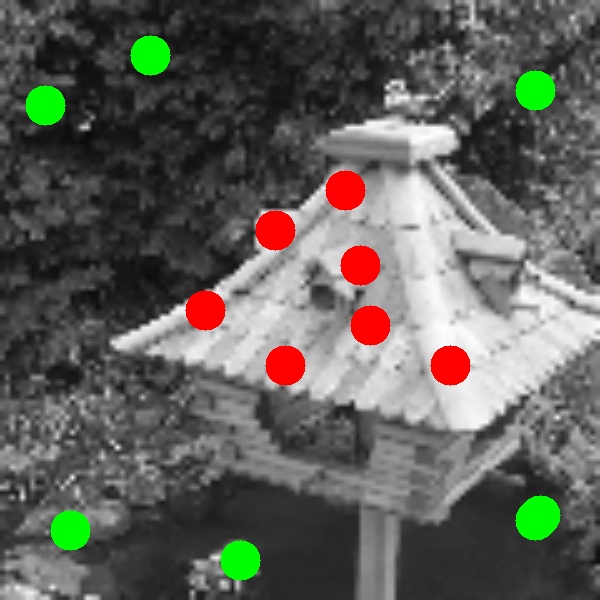}
            \caption[]%
            {{\small Image 10}}    
            \label{fig:bhseeds3}
        \end{subfigure}
        \caption{\small Seeds on the first, fifth, and last images from the $120 \times 120$ pixels \textsc{Birdhouse} sequence. Red for object, green for background.} 
        \label{fig:seed_example}
\end{figure}

\begin{figure}
        \centering
        \begin{subfigure}[b]{0.2\linewidth}
            \centering
            \includegraphics[width=\linewidth]{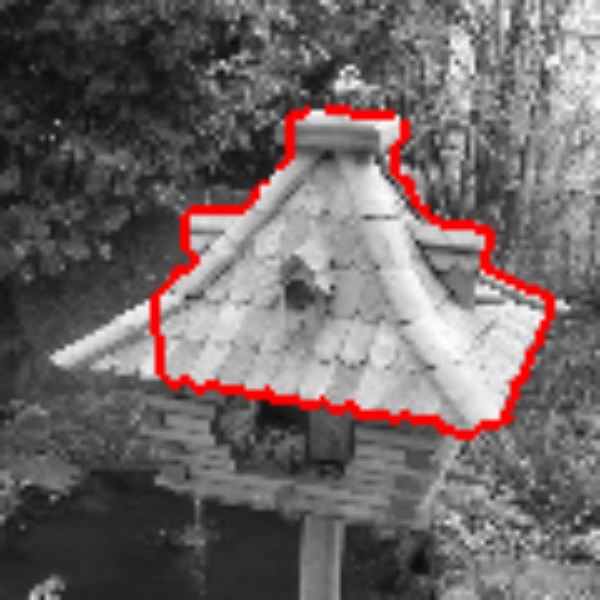}
            \caption[]%
            {{\small Image 1}}    
            \label{fig:bhcut1}
        \end{subfigure}
        \hspace{1cm}
        \begin{subfigure}[b]{0.2\linewidth}  
            \centering 
            \includegraphics[width=\linewidth]{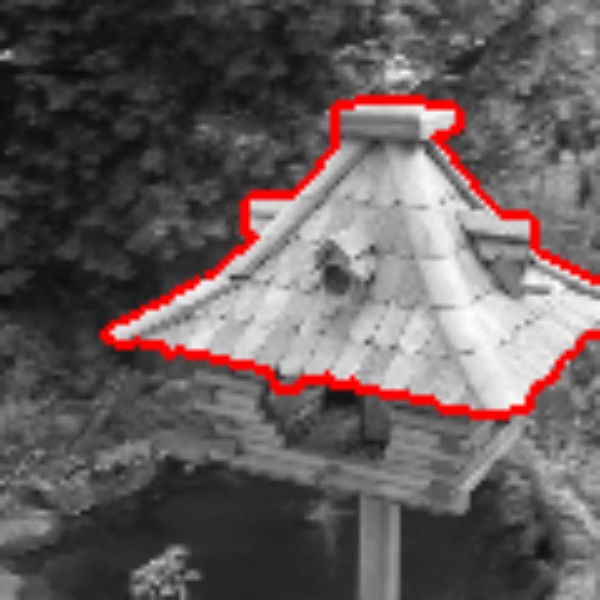}
            \caption[]%
            {{\small Image 5}}    
            \label{fig:bhcut2}
        \end{subfigure}
       \hspace{1cm}
            \begin{subfigure}[b]{0.2\linewidth}  
            \centering 
            \includegraphics[width=\linewidth]{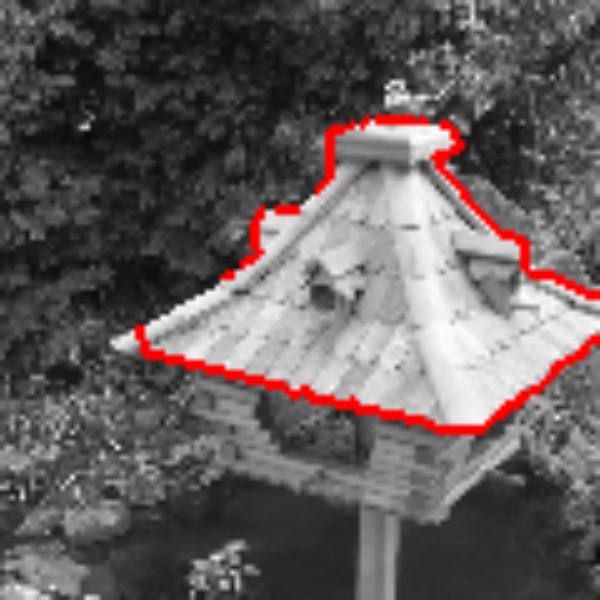}
            \caption[]%
            {{\small Image 10}}    
            \label{fig:bhcut3}
        \end{subfigure}
        \caption{\small  Cuts (red) on the first, fifth, and last images from the $120 \times 120$ pixels \textsc{Birdhouse} sequence.} 
        \label{fig:bh_cuts}
\end{figure}


We compare cold- and warm-start for both Edmonds-Karp and Dinic's algorithms. We use breadth-fist-search (BFS) to find such projection paths in 
our warm-starts for both Edmonds-Karp and Dinic's.
We use the BFS procedure for our warm-started Dinic's instead of the expected subroutine from Dinic's algorithm because the overhead of building the level graph is more  time consuming than running BFS.  This is due to the projection paths being short.


We use $n \times n$ pixel images for $n \in \{30, 60, 120\}$. 
Numerically, the $\sigma$ in the definition of $\beta_{\cdot, \cdot}$ is 50, and $C$ is $100$. 
To make the capacities integral, all $\beta_{p, q}$'s are rounded down to the nearest integer. 
Notice that $\beta_{p, q} \leq C$ by definition. 
We let $M=C|V|^2$ to make the term sufficiently large.

All experiments are run on a device with Intel(R) Core(TM) i7-7600U CPU @ 2.80GHz, with 24G memory. We record the wall-clock running time for both algorithms.
Many of the image process tools and functions are based on the \emph{Image Segmentation} Github project \citep{Julie2017}.

\begin{table}[ht]
\centering
\caption{Average running times of cold-/warm-start Ford Fulkerson and the percentage of time saved by warm-start, Edmonds-Karp}
\label{table:running_time}
\begin{tabular}{r|llll}
\hline
    Image Group & $30 \times 30$ & $60 \times 60$ & $120\times120$ \\ 
\hline
    \textsc{Birdhouse} & 0.83/0.55, 34.07\% & 8.48/3.48, 58.98\%  & 109.06/37.31, 65.78\% \\
    \textsc{Head} & 0.65/0.45 31.06\% & 9.52/4.28, 55.07\% & 112.66/31.77, 71.80\% \\
    \textsc{Shoe} & 0.72/0.46, 36.01\% & 8.81/3.04, 65.47\% & 111.05/30.44, 72.59\% \\
    \textsc{Dog} & 0.73/0.41, 42.96\% & 22.38/6.89, 69.22\%  & 202.99 / 42.04, 79.29\%\\
\hline
\end{tabular}
\end{table}

\begin{table}[ht]
\centering
\caption{Average running times of cold-/warm-start Ford Fulkerson and the percentage of time saved by warm-start, Dinic}
\label{table:running_time_dc}
\begin{tabular}{r|llll}
\hline
    Image Group & $30 \times 30$ & $60 \times 60$ & $120\times120$ \\ 
\hline
    \textsc{Birdhouse} & 0.38/0.37, 2.49\% & 5.81/3.17, 45.43\%  & 82.52/35.37, 57.14\% \\
    \textsc{Head} & 0.36/0.36 0.58\% & 7.7/4.44, 42.35\% & 149.12/49.44, 62.88\% \\
    \textsc{Shoe} & 0.39/0.37, 5.07\% & 7.01/3.35, 52.24\% & 140.52/49.33, 64.9\% \\
    \textsc{Dog} & 0.5/0.41, 10.16\% & 12.38/4.99, 59.66\%  & 206.85 / 58.98, 71.48\%\\
\hline
\end{tabular}
\end{table}

\noindent \textbf{Results}\quad
We first show that the boundary-based image segmentation approach generates reasonable cuts. 
For example, Figure \ref{fig:bh_cuts} illustrates cuts from the $120 \times 120$ \textsc{Birdhouse} sequence. See Appendix \ref{sec: appendix} for other examples.
We then compare the running time of cold- and warm-start Ford-Fulkerson. 
As all algorithms are returning optimal flows, 
there are no qualitative aspects of the solutions to measure. Tables \ref{table:running_time} and \ref{table:running_time_dc} show results in all experiments settings for Edmonds-Karp and Dinic, rows being image groups and columns image sizes. Each entry is formatted as ``cold-start time (s) / warm-start time(s), warm-start time savings (\%)''.

These results show warm-starting Ford-Fulkerson greatly improves the efficiency in all settings. Further, both cold- and warm- start's running time increases polynomially with the image width $n$, but warm-start grows slower, making it a potentially desirable approach on large scale networks. This is most obvious on image group \textsc{Dog} using Edmonds-Karp, where warm-start time is ~60\% of cold-start time on $30 \times 30$ pixels versus ~20\% on $120 \times 120$ pixels. These conclusions hold for both Edmonds-Karp and Dinic, with Dinic being slightly more efficient on smaller datasets.


Next we examine the execution of cold-/warm-start in more detail, taking the $120 \times 120$ \textsc{Birdhouse} sequence for Edmonds-Karp for example (Table \ref{table:path}). The table gives the average length of the augmenting paths (`avg length') and the average number of paths found (`avg \#') over the sequence of images. See Appendix \ref{sec: appendix} for complete data. 

\begin{table}[ht]
\centering
\caption{Comparison of projection and augmenting paths in cold- and warm-start Ford-Fulkerson, the first $5$ images from the $120 \times 120$ \textsc{Birdhouse} image sequence}
\label{table:path}
\scriptsize
\resizebox{\columnwidth}{!}{\begin{tabular}{r|llllll}
\hline Image \# & \tabincell{c}{cold-start \\ aug path \#}	& \tabincell{c}{cold-start \\ aug path \\ avg length} & \tabincell{c}{warm-start\\ proj path \#} &	\tabincell{c}{warm-start\\  proj path \\ avg length} & \tabincell{c}{warm-start\\ aug path \#} & \tabincell{c}{warm-start\\ aug path\\ avg length} \\
\hline 
1 & 2453 & 67.93 & 2105 & 9.39 & 628 & 81.48 \\
2 & 2093 & 65.22 & 3393 & 19.28 & 0 & 0 \\
3 & 2536 & 74.88 & 2038 & 9.71 & 896 & 101.731 \\
4 & 2089 & 69.09 & 3335 & 28.55 & 0 & 0 \\
5 & 1908 & 68.53 & 3226 & 22.97 & 0 & 0\\
\hline
\end{tabular}}
\end{table}

Results in Table \ref{table:path} suggest that the projected feasible flow is in general only slightly sub-optimal, which is key  for warm-start's efficiency. Max-flow on the previous image is  a good starting point for warm-start  with the feasibility projection algorithm. On average, after rounding down the previous max-flow to satisfy the new edge capacities, the total excess/deficit is $(1.75 \pm 0.44)$ \% of the real maximum flow value. Moreover,  fixing the  excess/deficit results in a near optimal flow.  Indeed, the projection quickly gives a feasible flow that recovers $(96 \pm 6)$\% of the maximum flow. 

Another  factor contributing to the efficiency of warm-start is the projection path-finding subroutine. Recall that both cold- and warm-start use the same BFS subroutine to find either an $s, t$ augmenting path or a projection path. The theory in Section \ref{sec: instance-robust} suggests that  paths in the projection step will take less time to find. To show this empirically, we collected data on the number of augmenting/feasibility projection paths found and their average lengths for both cold- and warm-start.  Overall, compared with cold-start, warm-start has shorter projection paths on average, suggesting massive savings in the BFS running time per path. While we show this for the \textsc{Birdhouse} 
images in Table \ref{table:path}, this is true on other datasets too, available in Appendix \ref{sec: appendix}. This explains the efficiency even if the  excess/deficit is large. This shows that the theoretical expectations raised in Section \ref{sec: instance-robust} are predictive of empirical performance.

\section{Conclusion}
We show how to warm-start the Ford-Fulkerson algorithm for computing flows, 
as well as prove strong theoretical results and give empirical evidence of good performance of our algorithm. 
We further refine our analysis to capture the gains due to using {\em short} projection paths to route excess flow and show that these scenarios are prevalent in image segmentation applications.

Many interesting challenges remain. 
For one, there are many known algorithms for computing flows, and it would be interesting to see if those methods can also be sped up in a similar fashion. A technical roadblock lies in handling both under- and over- predictions, particularly when predictions lead to infeasible flows. More generally, 
a network flow problem can be written as a linear program. Another direction is finding algorithms for solving general LPs that can be helped by judiciously chosen predictions.

\clearpage

\bibliographystyle{apalike}
\bibliography{refs}

\clearpage

\section{Appendix}
\label{sec: appendix}

We give a description of the experiment settings and provide more collected data and results. More can be found at \url{https://github.com/wang-yuyan/warmstart-graphcut-algorithms-pulic}.

\paragraph{More on choice of seeds and cuts} Recall that on the $10$ images from the same sequence, the seed pixels are always fixed. We note here the choice of seeds (number of seeds and their locations) affects which min-cut solution is found. 
However, as long as the seeds give a reasonable solution that is close to the real object/background boundary, the conclusions in the comparison between cold- and warm-start remain robust against a change of seeds. 

In Section \ref{sec: empirical}, we showed seeds and optimal cuts on images 1, 5, and 10 of the $120 \times 120$ pixel \textsc{Birdhouse} sequence in Figures \ref{fig:seed_example} and \ref{fig:bh_cuts}. Here we show, in addition, our seeds and resulting cuts on images 1, 5, and 10 of the other $120 \times 120$ sequences. Those of \textsc{Head} are in Figure \ref{fig:seed_cut_head}, those of \textsc{Shoe} in Figure \ref{fig:seed_cut_shoe}, 
and those of \textsc{Dog} in Figure \ref{fig:seed_cut_dog}.

\begin{figure}[ht]
        \centering
        \begin{subfigure}[b]{0.15\linewidth}
            \centering
            \includegraphics[width=\linewidth]{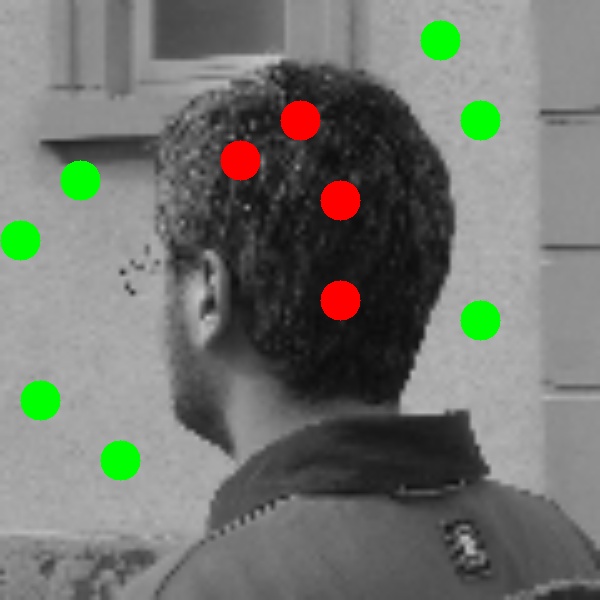}
            \caption[]%
            {{\footnotesize Image 1, seeds}}    
            \label{fig:headseed1}
        \end{subfigure}
        \hfill
        \begin{subfigure}[b]{0.15\linewidth}  
            \centering 
            \includegraphics[width=\linewidth]{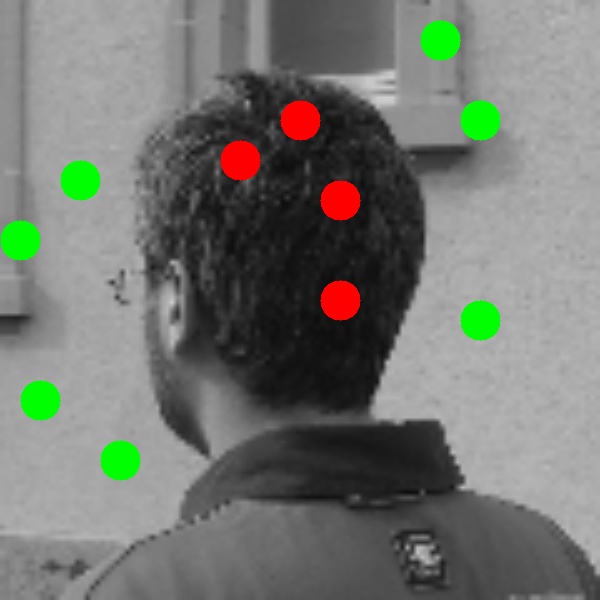}
            \caption[]%
            {{\footnotesize Image 5, seeds}}    
            \label{fig:headseed2}
        \end{subfigure}
        \hfill
            \begin{subfigure}[b]{0.15\linewidth}  
            \centering 
            \includegraphics[width=\linewidth]{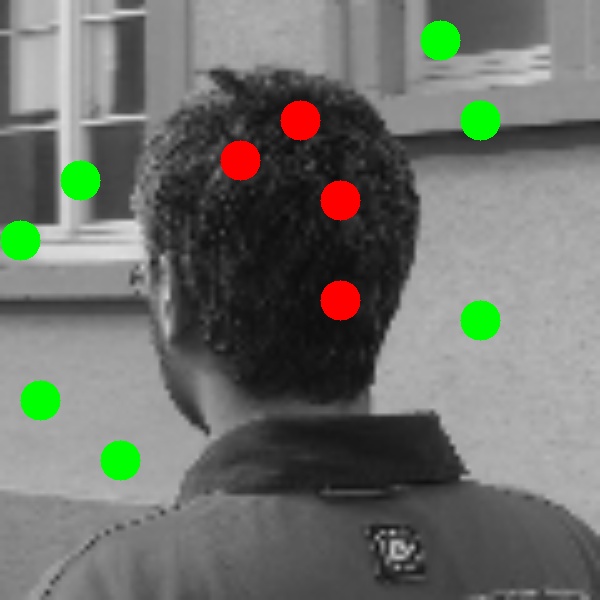}
            \caption[]%
            {{\footnotesize Image 10, seeds}}    
            \label{fig:headseed3}
        \end{subfigure}
        \hfill
        \begin{subfigure}[b]{0.15\linewidth}
            \centering
            \includegraphics[width=\linewidth]{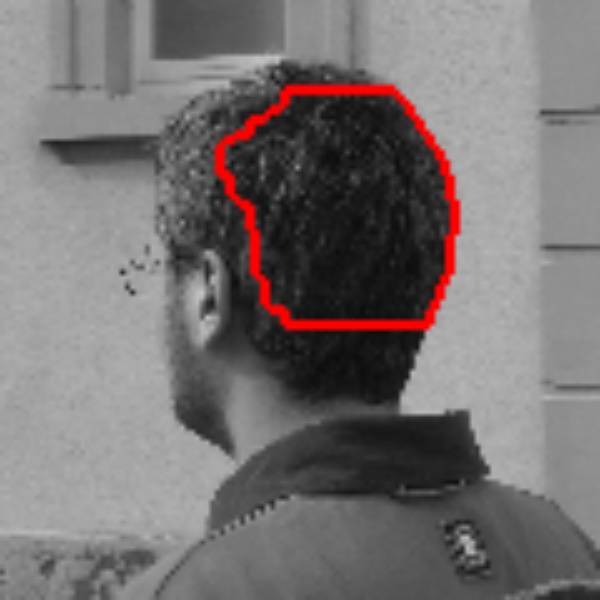}
            \caption[]%
            {{\footnotesize Image 1, cut}}    
            \label{fig:headcut1}
        \end{subfigure}
        \hfill
        \begin{subfigure}[b]{0.15\linewidth}  
            \centering 
            \includegraphics[width=\linewidth]{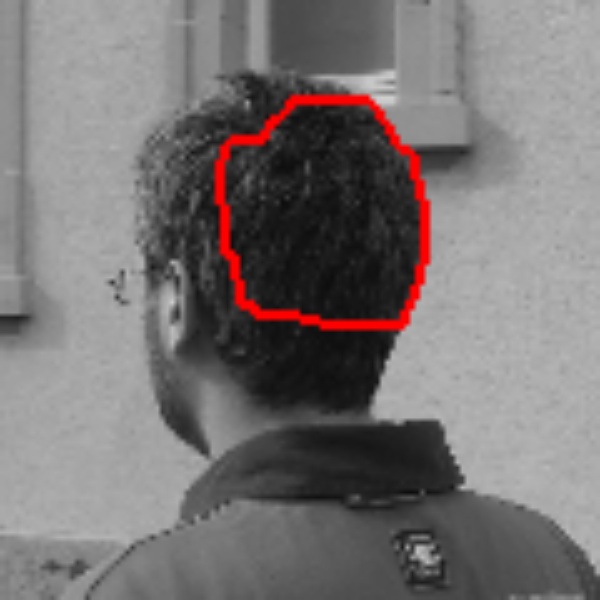}
            \caption[]%
            {{\footnotesize Image 5, cut}}    
            \label{fig:headcut2}
        \end{subfigure}
        \hfill
            \begin{subfigure}[b]{0.15\linewidth}  
            \centering 
            \includegraphics[width=\linewidth]{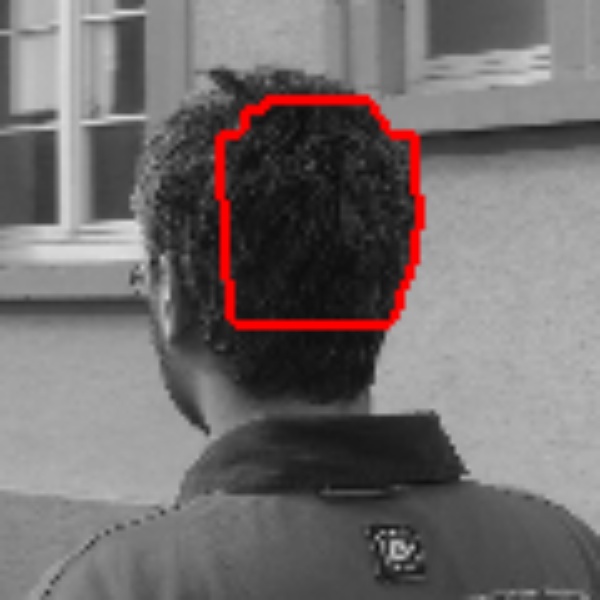}
            \caption[]%
            {{\footnotesize Image 10, cut}}    
            \label{fig:headcut3}
        \end{subfigure}
        \caption{\small Seeds and resulting cuts on the first, fifth, and last images from the $120 \times 120$ pixels \textsc{head} sequence. Red seeds for object, green seeds for background, red line for cut.} 
        \label{fig:seed_cut_head}
        \vspace{-0.1in}
\end{figure}

\begin{figure}[ht]
        \centering
        \begin{subfigure}[b]{0.15\linewidth}
            \centering
            \includegraphics[width=\linewidth]{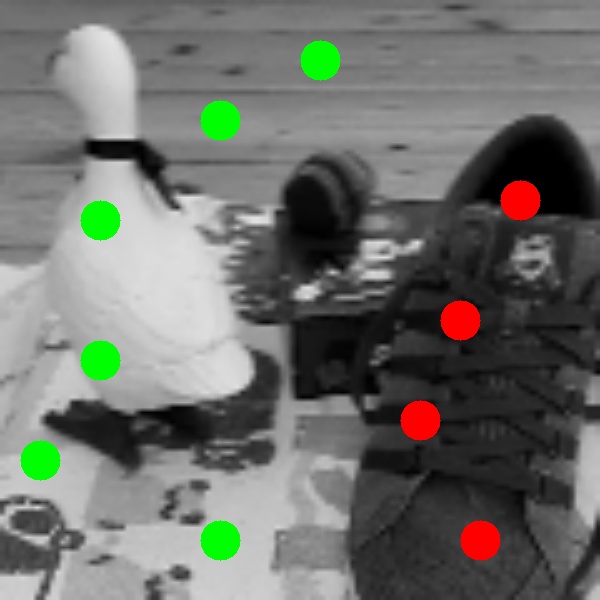}
            \caption[]%
            {{\footnotesize Image 1, seeds}}    
            \label{fig:shoeseed1}
        \end{subfigure}
        \hfill
        \begin{subfigure}[b]{0.15\linewidth}  
            \centering 
            \includegraphics[width=\linewidth]{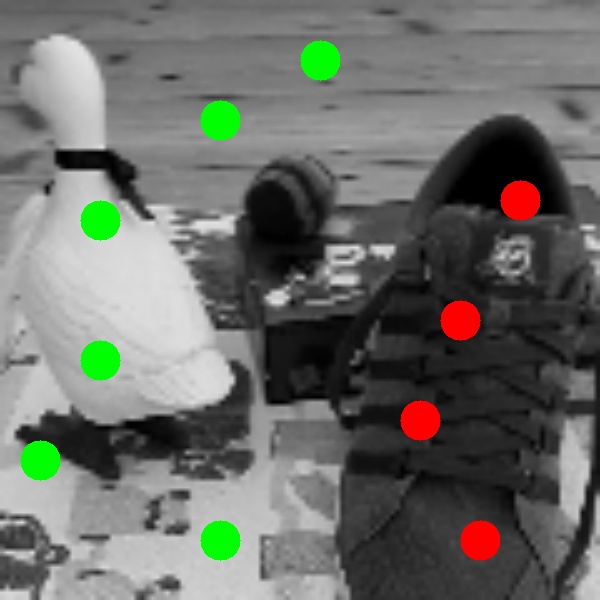}
            \caption[]%
            {{\footnotesize Image 5, seeds}}    
            \label{fig:shoeseed2}
        \end{subfigure}
        \hfill
            \begin{subfigure}[b]{0.15\linewidth}  
            \centering 
            \includegraphics[width=\linewidth]{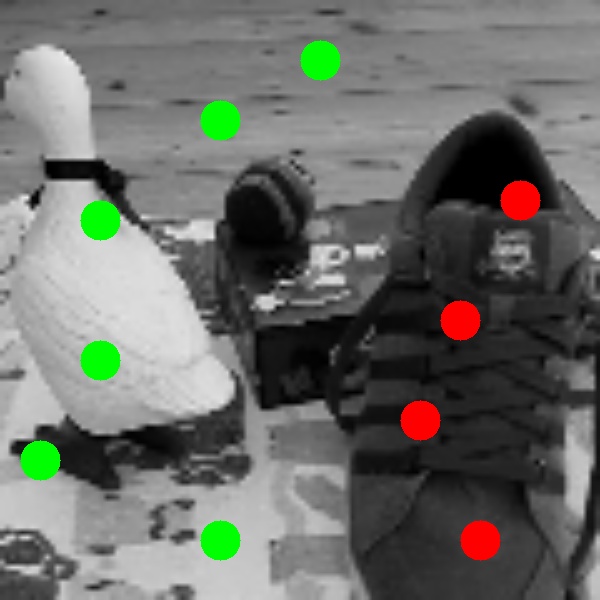}
            \caption[]%
            {{\footnotesize Image 10, seeds}}    
            \label{fig:shoeseed3}
        \end{subfigure}
        \hfill
        \begin{subfigure}[b]{0.15\linewidth}
            \centering
            \includegraphics[width=\linewidth]{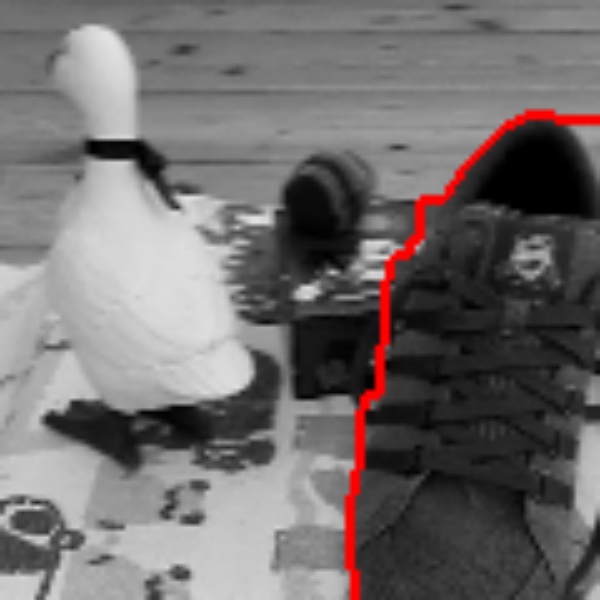}
            \caption[]%
            {{\footnotesize Image 1, cut}}    
            \label{fig:shoecut1}
        \end{subfigure}
        \hfill
        \begin{subfigure}[b]{0.15\linewidth}  
            \centering 
            \includegraphics[width=\linewidth]{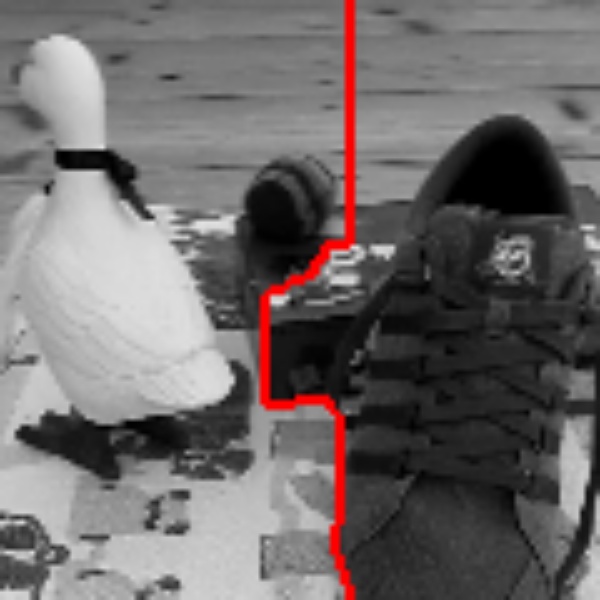}
            \caption[]%
            {{\footnotesize Image 5, cut}}    
            \label{fig:shoecut2}
        \end{subfigure}
        \hfill
            \begin{subfigure}[b]{0.15\linewidth}  
            \centering 
            \includegraphics[width=\linewidth]{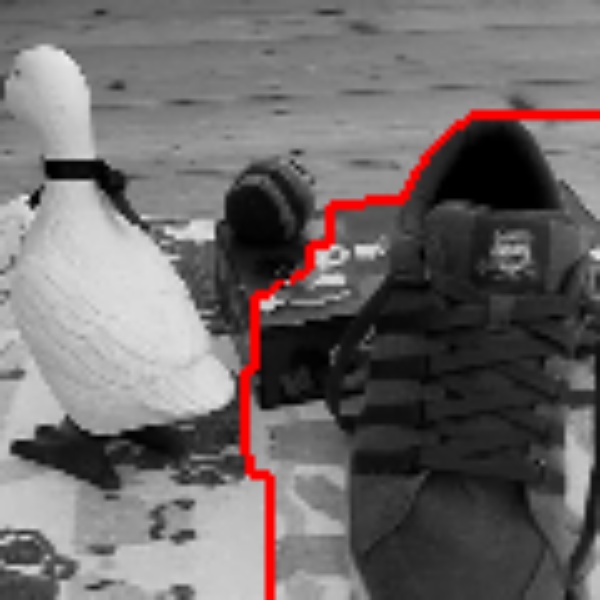}
            \caption[]%
            {{\footnotesize Image 10, cut}}    
            \label{fig:shoecut3}
        \end{subfigure}
        \caption{\small Seeds and resulting cuts on the first, fifth, and last images from the $120 \times 120$ pixels \textsc{shoe} sequence. Red seeds for object, green seeds for background, red line for cut.} 
        \label{fig:seed_cut_shoe}
        \vspace{-0.1in}
\end{figure}

\begin{figure}[ht]
        \centering
        \begin{subfigure}[b]{0.15\linewidth}
            \centering
            \includegraphics[width=\linewidth]{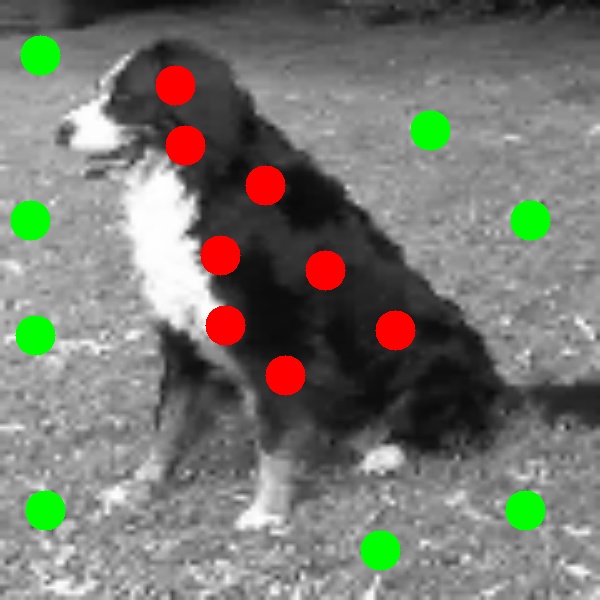}
            \caption[]%
            {{\footnotesize Image 1, seeds}}    
            \label{fig:dogseed1}
        \end{subfigure}
        \hfill
        \begin{subfigure}[b]{0.15\linewidth}  
            \centering 
            \includegraphics[width=\linewidth]{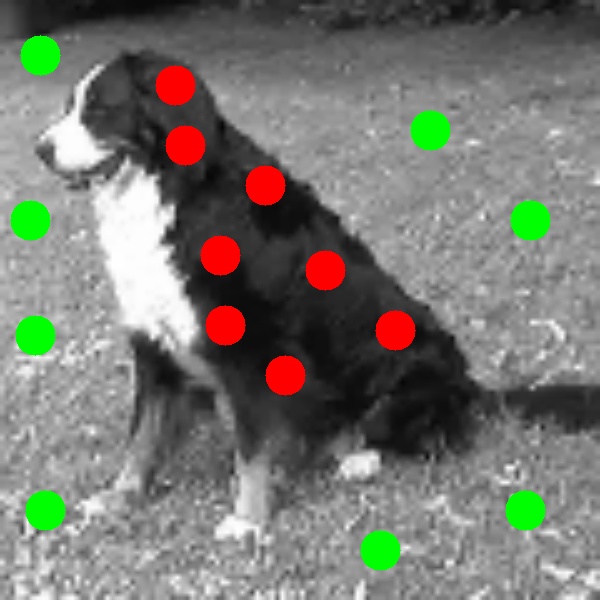}
            \caption[]%
            {{\footnotesize Image 5, seeds}}    
            \label{fig:dogseed2}
        \end{subfigure}
        \hfill
            \begin{subfigure}[b]{0.15\linewidth}  
            \centering 
            \includegraphics[width=\linewidth]{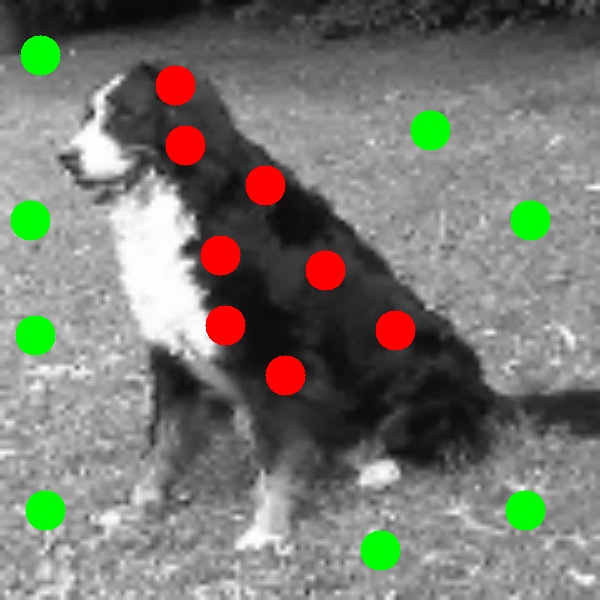}
            \caption[]%
            {{\footnotesize Image 10, seeds}}    
            \label{fig:dogseed3}
        \end{subfigure}
        \hfill
        \begin{subfigure}[b]{0.15\linewidth}
            \centering
            \includegraphics[width=\linewidth]{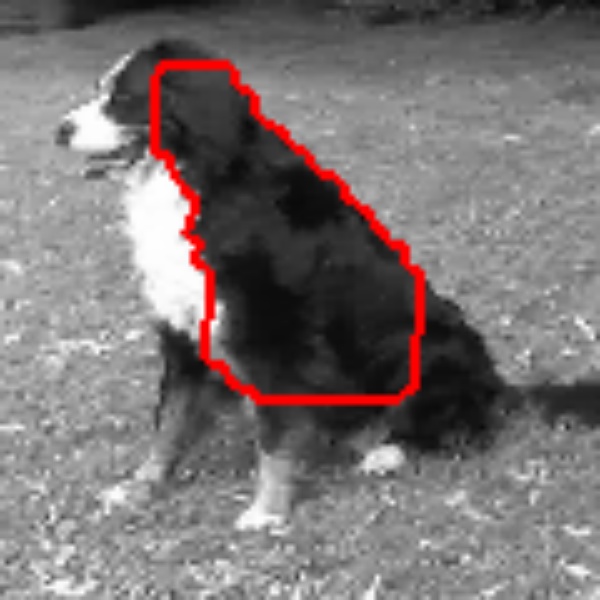}
            \caption[]%
            {{\footnotesize Image 1, cut}}    
            \label{fig:dogcut1}
        \end{subfigure}
        \hfill
        \begin{subfigure}[b]{0.15\linewidth}  
            \centering 
            \includegraphics[width=\linewidth]{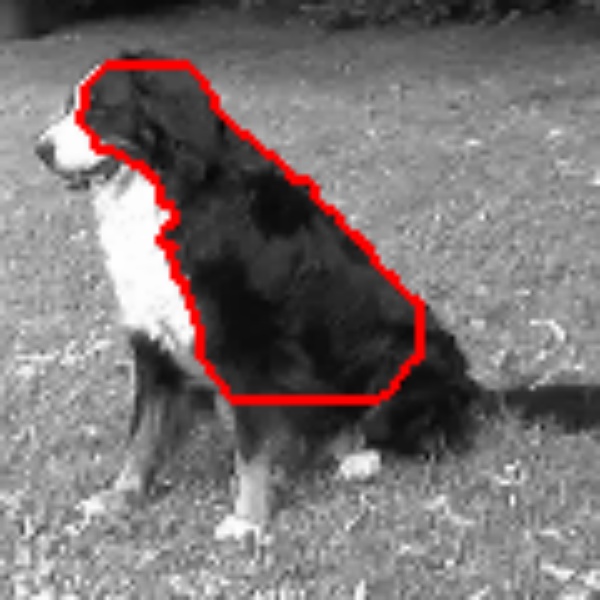}
            \caption[]%
            {{\footnotesize Image 5, cut}}    
            \label{fig:dogcut2}
        \end{subfigure}
        \hfill
            \begin{subfigure}[b]{0.15\linewidth}  
            \centering 
            \includegraphics[width=\linewidth]{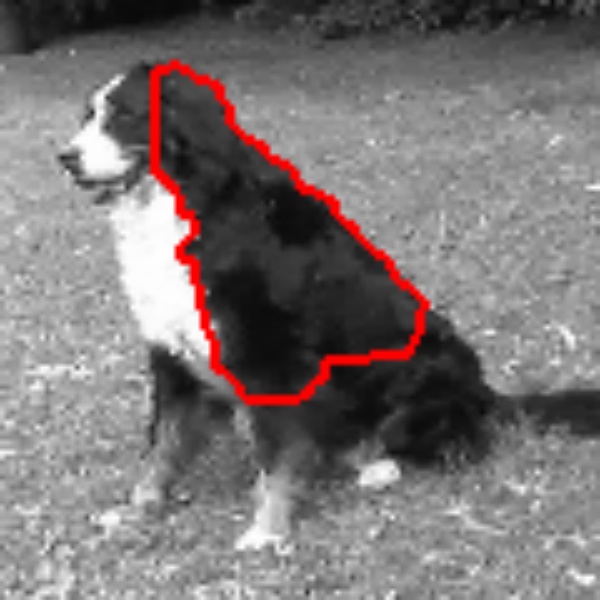}
            \caption[]%
            {{\footnotesize Image 10, cut}}    
            \label{fig:dogcut3}
        \end{subfigure}
        \caption{\small Seeds and resulting cuts on the first, fifth, and last images from the $120 \times 120$ pixels \textsc{dog} sequence. Red seeds for object, green seeds for background, red line for cut.} 
        \label{fig:seed_cut_dog}
        \vspace{-0.1in}
\end{figure}

When we select a seed, we draw a two-dimensional ball around the target seed and let every pixel in this ball be a seed as well. We found this practice to work better than simply choosing individual pixels as seeds. When we switch from low-resolution ($30 \times 30$) to high-resolution ($120 \times 120$) images, we rescale the radius of this ball proportional to the number of pixels on each side. On the $30 \times 30$, $60 \times 60$, $120 \times 120$ pixel images, the ball's radius is $1$, $2$ and $4$ pixels, respectively. In other words, if we stretch/compress the images of different resolution to be the same size, the ball will roughly have the same area geometrically. We also found this to be more effective than fixing the pixel radius, despite the change in resolution.

\begin{figure}[ht]
        \centering
        \begin{subfigure}[b]{0.15\linewidth}
            \centering
            \includegraphics[width=\linewidth]{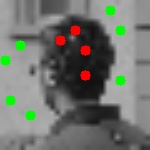}
            \caption[]%
            {{\small $30 \times 30$}}    
            \label{fig:headseed30}
        \end{subfigure}
        \hspace{1cm}
        \begin{subfigure}[b]{0.15\linewidth}  
            \centering 
            \includegraphics[width=\linewidth]{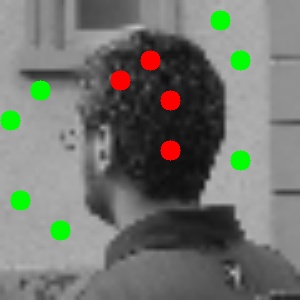}
            \caption[]%
            {{\small $60 \times 60$}}    
            \label{fig:headseed60}
        \end{subfigure}
        \hspace{1cm}
            \begin{subfigure}[b]{0.15\linewidth}  
            \centering 
            \includegraphics[width=\linewidth]{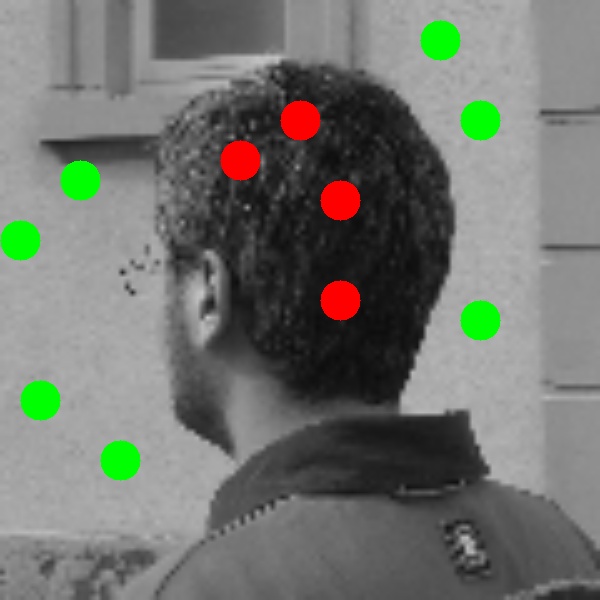}
            \caption[]%
            {{\small $120 \times 120$}}    
            \label{fig:headseed120}
        \end{subfigure}
        \caption{\small Area that each seed covers on the same image with different resolutions. } 
        \label{fig:head_seed_sizes}
        \vspace{-0.1in}
\end{figure}

Note that although the location of the seeds remains unchanged throughout an image sequence,  we may still need to provide more seeds when we switch from low- to high-resolution images. Intuitively, blurring the image lessens the minor contrast of pixels within the object and makes the geometric shape easier to capture.
The seeds and min-cut results on the $30 \times 30$ and $60 \times 60$ sequences can be found in the code directory uploaded in the github linked at the beginning of this section.

\paragraph{More on the warm-start magic} In the main body we gave evidence---both theoretically and empirically---that the savings in the run-time of warm-start is mostly due to:
\begin{itemize}
\item The algorithm's ability to use short projection paths to re-route excess flow to nodes with deficit flow, thus projecting the predicted flow to a feasible one quickly.
\item An only slightly sub-optimal flow after the feasibility projection, so that warm-start takes fewer augmenting paths to reach an optimal flow.
\end{itemize}

Here we provide more results in support of these two claims.
To show the level of total excess/deficit (whichever one is larger) and the flow value after the feasibility projection step, we show two ratios: total excess/deificit over max-flow (Table \ref{table:excess_ratio}), and feasible flow value over max-flow (Table \ref{table:flow_ratio}). One can see that typically the total excess/deficit is not negligible. In fact they are quite high and if the algorithm does not resolve excesses/deficits in the right way (such as sending all excess to the source) it could cause the flow value to diminish a lot. Our feasibility projection makes good decisions about using projection paths to make up for excess/deficit, so that it outputs a feasible flow with almost optimal flow value.

\begin{table}[ht]
\centering
\caption{Average ratio of total excess/deficit over max-flow value in warm-start}
\label{table:excess_ratio}
\footnotesize
\begin{tabular}{r|llll}
\hline
    Image Group & $30 \times 30$ & $60 \times 60$ & $90\times90$ \\ 
\hline
    \textsc{Birdhouse} & 1.06 $\pm$ 0.22 & 1.60 $\pm$ 0.21 & 1.75 $\pm$ 0.44 \\
    \textsc{Head} & 0.49 $\pm$ 0.12 & 0.6 $\pm$ 0.12 & 0.74 $\pm$ 0.1 \\
    \textsc{Shoe} & 0.49 $\pm$ 0.13 & 0.66 $\pm$ 0.08 & 0.95 $\pm$ 0.14\\
    \textsc{Dog} & 0.55 $\pm$ 0.07 & 0.8 $\pm$ 0.07 & 1.08 $\pm$ 0.19 \\
\hline
\end{tabular}
\end{table}

\begin{table}[ht]
\centering
\caption{Average ratio of flow value after feasibility projection over max-flow value  in warm-start}
\label{table:flow_ratio}
\footnotesize
\begin{tabular}{r|llll}
\hline
    Image Group & $30 \times 30$ & $60 \times 60$ & $90\times90$ \\ 
\hline
    \textsc{Birdhouse} & 0.94 $\pm$ 0.09 & 0.98 $\pm$ 0.03 & 0.96 $\pm$ 0.06 \\
    \textsc{Head} & 0.98 $\pm$ 0.03 & 0.98 $\pm$ 0.03 & 0.99 $\pm$ 0.01 \\
    \textsc{Shoe} & 0.98 $\pm$ 0.02 & 0.98 $\pm$ 0.03 & 0.98 $\pm$ 0.02\\
    \textsc{Dog} & 0.97 $\pm$ 0.04 & 0.97 $\pm$ 0.03 & 0.98 $\pm$ 0.03 \\
\hline
\end{tabular}
\end{table}

To show that the conclusion of projection paths being short broadly holds for all image groups, we give the average length of the augmenting and projection paths (`avg length') and the number of paths found (`aug path \#' and `proj path \#') over the first 5 images in the sequence for the $120 \times 120$ \textsc{Head} sequence in Table \ref{table:path-head}, 
the $120 \times 120$ \textsc{Shoe} sequence in Table \ref{table:path-shoe}, 
and the $120 \times 120$ 
\textsc{Dog} sequence in Table \ref{table:path-dog}.
Note the analogous table for the $120 \times 120$ \textsc{Birdhouse} 
sequence (Table \ref{table:path}) is in Section \ref{sec: empirical}.

\begin{table}[h!]
\centering
\caption{Comparison of projection and augmenting paths in cold- and warm-start Ford-Fulkerson, the first $5$ images from the $120 \times 120$ \textsc{Head} image sequence}
\label{table:path-head}
\footnotesize
\begin{tabular}{r|llllll}
\hline Image \# & \tabincell{c}{cold-start \\ aug path \#}	& \tabincell{c}{cold-start \\ aug path \\ avg length} & \tabincell{c}{warm-start\\ proj path \#} &	\tabincell{c}{warm-start\\  proj path \\ avg length} & \tabincell{c}{warm-start\\ aug path \#} & \tabincell{c}{warm-start\\ aug path\\ avg length} \\
\hline 
1 & 2714 & 82.65 & 2573 & 15.93 & 221 & 80.42 \\
2 & 2687 & 82.74 & 2512 & 20.40 & 217 & 135.68 \\
3 & 2475 & 76.63 & 2667 & 19.78 & 0 & 0 \\
4 & 2379 & 76.44 & 2140 & 17.00 & 0 & 0 \\
5 & 2349 & 75.66 & 2260 & 19.97 & 112 & 138.14\\
\hline
\end{tabular}
\end{table}

\begin{table}[h!]
\centering
\caption{Comparison of projection and augmenting paths in cold- and warm-start Ford-Fulkerson, the first $5$ images from the $120 \times 120$ \textsc{Shoe} image sequence}
\label{table:path-shoe}
\footnotesize
\begin{tabular}{r|llllll}
\hline Image \# & \tabincell{c}{cold-start \\ aug path \#}	& \tabincell{c}{cold-start \\ aug path \\ avg length} & \tabincell{c}{warm-start\\ proj path \#} &	\tabincell{c}{warm-start\\  proj path \\ avg length} & \tabincell{c}{warm-start\\ aug path \#} & \tabincell{c}{warm-start\\ aug path\\ avg length} \\
\hline 
1 & 1948 & 89.23 & 2252 & 22.70 & 0 & 0 \\
2 & 2081 & 91.67 & 1992 & 16.54 & 112 & 148.41 \\
3 & 2039 & 93.88 & 1936 & 14.91 & 177 & 142.51 \\
4 & 2110 & 101.97 & 2525 & 35.04 & 0 & 0 \\
5 & 2016 & 93.68 & 2375 & 18.60 & 0 & 0\\
\hline
\end{tabular}
\end{table}

\begin{table}[h!]
\centering
\caption{Comparison of projection and augmenting paths in cold- and warm-start Ford-Fulkerson, the first $5$ images from the $120 \times 120$ \textsc{Dog} image sequence}
\label{table:path-dog}
\footnotesize
\begin{tabular}{r|llllll}
\hline Image \# & \tabincell{c}{cold-start \\ aug path \#}	& \tabincell{c}{cold-start \\ aug path \\ avg length} & \tabincell{c}{warm-start\\ proj path \#} &	\tabincell{c}{warm-start\\  proj path \\ avg length} & \tabincell{c}{warm-start\\ aug path \#} & \tabincell{c}{warm-start\\ aug path\\ avg length} \\
\hline 
1 & 3314 & 63.04 & 3684 & 12.51 & 0 & 0 \\
2 & 3200 & 65.56 & 4611 & 21.69 & 0 & 0 \\
3 & 3138 & 63.53 & 3515 & 12.30 & 0 & 0 \\
4 & 3259 & 66.61 & 3270 & 10.74 & 444 & 87.08 \\
5 & 3120 & 64.43 & 3932 & 12.63 & 0 & 0\\
\hline
\end{tabular}
\end{table}

Further, we show the equivalence of these tables for the other two image sizes/resolutions, $30 \times 30$ and $60 \times 60$, for image groups \textsc{head} (Table \ref{table:path-30-head} and \ref{table:path-60-head}) and \textsc{shoe} (Table \ref{table:path-30-shoe} and \ref{table:path-60-shoe}). For these two groups, sequences of all three sizes share the same location of seeds. One can see that, the average length of the augmenting paths in cold-start Ford-Fulkerson grows roughly proportional to the width of the image. The average length of the projection paths during the warm-start feasibility projection also grows as the width of the image grows, but slightly slower than the former. This could potentially cause warm-start to be more advantageous on high-resolution images.

The omitted data tables and other experiment results can be found in the uploaded program directory (see the README.md file in the linked github repository for instructions.
\begin{table}[h]
\centering
\caption{Comparison of projection and augmenting paths in cold- and warm-start Ford-Fulkerson, the first $5$ images from the $30 \times 30$ \textsc{Head} image sequence}
\label{table:path-30-head}
\footnotesize
\begin{tabular}{r|llllll}
\hline Image \# & \tabincell{c}{cold-start \\ aug path \#}	& \tabincell{c}{cold-start \\ aug path \\ avg length} & \tabincell{c}{warm-start\\ proj path \#} &	\tabincell{c}{warm-start\\  proj path \\ avg length} & \tabincell{c}{warm-start\\ aug path \#} & \tabincell{c}{warm-start\\ aug path\\ avg length} \\
\hline 
1 & 267 & 25.16 & 226 & 8.61 & 61 & 40.72 \\
2 & 244 & 23.26 & 254 & 11.63 & 3 & 44.33 \\
3 & 253 & 22.11 & 236 & 12.05 & 0 & 0 \\
4 & 248 & 21.45 & 238 & 11.17 & 0 & 0 \\
5 & 250 & 22.98 & 252 & 12.24 & 10 & 43.30\\
\hline
\end{tabular}
\end{table}

\begin{table}[h!]
\centering
\caption{Comparison of projection and augmenting paths in cold- and warm-start Ford-Fulkerson, the first $5$ images from the $60 \times 60$ \textsc{Head} image sequence}
\label{table:path-60-head}
\footnotesize
\begin{tabular}{r|llllll}
\hline Image \# & \tabincell{c}{cold-start \\ aug path \#}	& \tabincell{c}{cold-start \\ aug path \\ avg length} & \tabincell{c}{warm-start\\ proj path \#} &	\tabincell{c}{warm-start\\  proj path \\ avg length} & \tabincell{c}{warm-start\\ aug path \#} & \tabincell{c}{warm-start\\ aug path\\ avg length} \\
\hline 
1 & 789 & 46.52 & 674 & 10.36 & 164 & 56.57 \\
2 & 852 & 44.17 & 763 & 9.69 & 99 & 62.59 \\
3 & 752 & 41.09 & 866 & 14.71 & 0 & 0 \\
4 & 782 & 40.19 & 567 & 7.52 & 169 & 48.0 \\
5 & 777 & 42.62 & 931 & 16.67 & 0 & 0\\
\hline
\end{tabular}
\end{table}

\begin{table}[h!]
\centering
\caption{Comparison of projection and augmenting paths in cold- and warm-start Ford-Fulkerson, the first $5$ images from the $30 \times 30$ \textsc{Shoe} image sequence}
\label{table:path-30-shoe}
\footnotesize
\begin{tabular}{r|llllll}
\hline Image \# & \tabincell{c}{cold-start \\ aug path \#}	& \tabincell{c}{cold-start \\ aug path \\ avg length} & \tabincell{c}{warm-start\\ proj path \#} &	\tabincell{c}{warm-start\\  proj path \\ avg length} & \tabincell{c}{warm-start\\ aug path \#} & \tabincell{c}{warm-start\\ aug path\\ avg length} \\
\hline 
1 & 165 & 20.85 & 192 & 9.59 & 0 & 0 \\
2 & 172 & 20.72 & 164 & 7.53 & 24 & 27.21 \\
3 & 175 & 21.91 & 195 & 12.42 & 2 & 34.5 \\
4 & 201 & 22.51 & 164 & 12.34 & 15 & 29.93 \\
5 & 162 & 21.21 & 215 & 9.22 & 0 & 0\\
\hline
\end{tabular}
\end{table}

\begin{table}[ht]
\centering
\caption{Comparison of projection and augmenting paths in cold- and warm-start Ford-Fulkerson, the first $5$ images from the $60 \times 60$ \textsc{Shoe} image sequence}
\label{table:path-60-shoe}
\footnotesize
\begin{tabular}{r|llllll}
\hline Image \# & \tabincell{c}{cold-start \\ aug path \#}	& \tabincell{c}{cold-start \\ aug path \\ avg length} & \tabincell{c}{warm-start\\ proj path \#} &	\tabincell{c}{warm-start\\  proj path \\ avg length} & \tabincell{c}{warm-start\\ aug path \#} & \tabincell{c}{warm-start\\ aug path\\ avg length} \\
\hline 
1 & 585 & 41.22 & 580 & 13.40 & 31 & 58.65 \\
2 & 508 & 40.21 & 562 & 14.06 & 0 & 0 \\
3 & 609 & 42.29 & 469 & 7.13 & 147 & 50.65 \\
4 & 646 & 43.86 & 675 & 14.13 & 17 & 58.24 \\
5 & 595 & 44.64 & 683 & 14.82 & 0 & 0\\
\hline
\end{tabular}
\end{table}

\end{document}